\PassOptionsToPackage{dvipsnames,svgnames,table}{xcolor}
\documentclass[conference]{IEEEtran}
\IEEEoverridecommandlockouts

\def\VersionWithComments

\usepackage{cite}
\usepackage{amsmath,amssymb,amsfonts}
\usepackage{algorithmic}
\usepackage{graphicx}
\usepackage{textcomp}
\usepackage{booktabs}
\usepackage{xcolor}
\def\BibTeX{{\rm B\kern-.05em{\sc i\kern-.025em b}\kern-.08em
    T\kern-.1667em\lower.7ex\hbox{E}\kern-.125emX}}
\usepackage{colortbl}

\usepackage{balance} 
\usepackage{xspace}
\usepackage[super]{nth}  
\usepackage{tikz}
\usepackage{subcaption}
\usepackage{eurosym}
\usepackage[normalem]{ulem}
\usepackage{diagbox}
\usepackage{pifont}
\usepackage{mathtools}
\usepackage[fillcomment,nosemicolon,ruled,vlined,linesnumbered]{algorithm2e}
\usepackage{etoolbox}
\usepackage[hidelinks]{hyperref}
\usepackage[capitalise,english,nameinlink]{cleveref}
\usepackage{amsthm}
\usepackage[inline,shortlabels]{enumitem}

\newtheorem{example}{Example}
\newtheorem{theorem}{Theorem}
\newtheorem{proposition}{Proposition}

\crefname{algorithm}{alg.}{algs.}
\Crefname{algorithm}{Algorithm}{Algorithms}

\SetKwComment{Comment}{$\triangleright$\,}{}
\SetVlineSkip{0pt}

\SetCommentSty{algcomment}

\makeatletter
\patchcmd{\@algocf@start}
  {-1.5em}
  {0pt}
  {}{}
\makeatother

\usetikzlibrary{calc}
\usetikzlibrary{arrows,shapes,automata,backgrounds,petri,intersections}
\usetikzlibrary{decorations}
\usetikzlibrary{decorations.pathmorphing}
\usetikzlibrary{shadows}
\usetikzlibrary{shapes.symbols}
\usetikzlibrary{shapes.callouts}
\usetikzlibrary{shapes.geometric}
\usetikzlibrary{patterns}
\usetikzlibrary{positioning}
\usetikzlibrary{backgrounds}
\usetikzlibrary{fit,patterns,shapes.multipart,calc,matrix,shapes.gates.logic,shapes.gates.logic.US,circuits}

\tikzstyle{place}=[circle,thick,draw=blue!75,fill=blue!20,minimum size=6mm]
\tikzstyle{contour place}=[place,draw=red!100]
\tikzstyle{red place}=[place,draw=red!75,fill=red!20]
\tikzstyle{gray place}=[place,draw=black!100,fill=black!30]
\tikzstyle{transition}=[ rectangle,thick, fill=black, minimum width=8mm, inner ysep=2pt]
\tikzstyle{red transition}=[ rectangle,thick, fill=red!75, minimum width=8mm, inner ysep=2pt]
\tikzstyle{gray transition}=[ rectangle,thick, draw=black!100,fill=black!30, minimum width=8mm, inner ysep=2pt]
\tikzstyle{blue transition}=[ rectangle,thick, fill=blue!75, minimum width=8mm, inner ysep=2pt]
\tikzstyle{every label}=[black]
\tikzstyle{gateSEQ}=[diamond]
\tikzstyle{gateNULL}=[trapezium,trapezium left angle=120, trapezium right angle=120]
\tikzset{depth label/.style={
  label={[draw=none,green!60!black,label distance=-2pt]left:#1}}}
\tikzset{level label/.style={
  label={[draw=none,blue,label distance=-4pt]right:#1}}}
\tikzstyle{every node}=[initial text=]
\tikzstyle{location}=[rectangle, rounded corners, minimum size=12pt, draw=black, fill=blue!10, inner sep=2pt]
\tikzstyle{location10}=[location, minimum size=10pt]
\tikzstyle{invariant}=[draw=black, dotted, inner sep=1pt] 

\tikzstyle{goodtile} = [draw=green!50!black, fill=green]
\tikzstyle{badtile} = [draw=red!50!black, fill=red]

\tikzstyle{invisible}=[draw=none]
\tikzstyle{final}=[double]
\tikzstyle{urgent}=[fill=yellow!50]
\tikzstyle{bad}=[fill=red!50]

\tikzstyle{every node}=[initial text=]
\tikzstyle{final}=[double]
\tikzstyle{sync}=[draw=blue,thick]
\tikzstyle{square}=[regular polygon,regular polygon sides=4]

\tikzstyle{seq}=[path picture={
	\draw[->]  ([xshift=#1, yshift=6pt]path picture bounding box.south)
					-- ([xshift=#1, yshift=-7pt]path picture bounding box.south);}]

\usepackage{amsmath} 
\usepackage{amssymb} 

\ifdefined \VersionWithComments
  \usepackage{marginnote}
  \newcommand{\marginX}{\marginnote{\huge{\quad\textbf{!}\quad}}}
  \newcommand{\wop}[1]{\textcolor{red!50!blue}{\marginX{}[\textbf{Wojciech}: #1
      ]}}
  \newcommand{\lp}[1]{\textcolor{blue!50}{\marginX{}[\textbf{Laure}: #1 ]}}
  \newcommand{\teo}[1]{\textcolor{green!50!black}{\marginX{}[\textbf{Teofil}: #1]}}
  \newcommand{\ja}[1]{\textcolor{DarkOrchid}{\marginX{}[\textbf{Jaime}: #1 ]}}
  \newcommand{\co}[1]{\textcolor{orange}{\marginX{}[\textbf{Carlos}: #1 ]}}
  \newcommand{\todo}[1]{\textcolor{red}{\marginX{}TODO: #1}}
\else
  \newcommand{\wop}[1]{}
  \newcommand{\lp}[1]{}
  \newcommand{\ja}[1]{}
  \newcommand{\co}[1]{}
  \newcommand{\teo}[1]{}
  \newcommand{\todo}[1]{}
\fi

\newcommand{\eg}{\emph{e.g.}\xspace}
\newcommand{\ie}{\emph{i.e.}\xspace}

\newcommand{\tool}{\texttt{ADT2AMAS}\xspace}

\def\ADT/{ADTree}
\def\DAG/{DAG}

\newcommand{\actAttack}{\ensuremath{a}}
\newcommand{\ActAttack}{\ensuremath{A}}

\newcommand{\actDefence}{\ensuremath{d}}

\newcommand{\leaf}[1]{\ensuremath{\mathtt{#1}}\xspace}
\newcommand{\gate}[1]{\ensuremath{\mathtt{\MakeUppercase{#1}}}\xspace}
\newcommand{\gateAND}{\gate{and}}
\newcommand{\gateOR}{\gate{or}}
\newcommand{\gateCAND}{\gate{cand}}

\newcommand{\gateLEAF}{\gate{leaf}}
\newcommand{\gateNULL}{\gate{null}}
\newcommand{\gateNODEF}{\gate{nodef}}
\newcommand{\gateSAND}{\gate{sand}}
\newcommand{\gateSEQ}{\gate{seq}}

\newcommand{\gateSCAND}{\gate{scand}}

\newcommand{\tunit}{\ensuremath{t_{\mathit{unit}}}}
\newcommand{\DAGset}{\ensuremath{\mathit{DAG\_set}}}
\newcommand{\inputDAG}{\ensuremath{\mathit{DAG}}}
\newcommand{\workingSet}{\ensuremath{S}}
\newcommand{\N}{\ensuremath{N}}

\newcommand{\varAgent}{\ensuremath{\mathit{agent}}}
\newcommand{\varSlot}{\ensuremath{\mathit{slot}}}
\newcommand{\varSlots}{\ensuremath{\mathit{slots}}}
\newcommand{\varChildNode}{\ensuremath{\mathit{child\_node}}}
\newcommand{\varCurrentNode}{\ensuremath{\mathit{current\_node}}}
\newcommand{\varParentAgent}{\ensuremath{\mathit{par\_agent}}}
\newcommand{\varParentNode}{\ensuremath{\mathit{parent\_node}}}
\newcommand{\varNodesLeft}{\ensuremath{\mathit{n\_remain}}}
\newcommand{\varBound}{\ensuremath{\mathit{low\_bound}}}
\newcommand{\varUpperBound}{\ensuremath{\mathit{up\_bound}}}
\newcommand{\varMaxAgents}{\ensuremath{\mathit{max\_agents}}}
\newcommand{\varNumAgents}{\ensuremath{\mathit{agents}}}
\newcommand{\varOutput}{\ensuremath{\mathit{output}}}
\newcommand{\varCurrentOutput}{\ensuremath{\mathit{curr\_output}}}
\newcommand{\varCandidate}{\ensuremath{\mathit{candidate}}}

\newcommand{\agent}{\ensuremath{\mathit{agent}}}

\newcommand{\children}{\ensuremath{\mathit{child}}}
\newcommand{\depth}{\ensuremath{\mathit{depth}}}

\newcommand{\keep}{\ensuremath{\mathit{keep}}}
\newcommand{\level}{\ensuremath{\mathit{level}}}
\newcommand{\node}{\ensuremath{\mathit{node}}}
\newcommand{\parent}{\ensuremath{\mathit{parent}}}
\newcommand{\slot}{\ensuremath{\mathit{slot}}}
\newcommand{\type}{\ensuremath{\mathit{type}}}

\usepackage{listings}
\usepackage{xcolor}

\usepackage{mycommands}

\begin{document}

\title{Optimal Scheduling of Agents in \ADT/s: Specialised Algorithm and Declarative Models
  \thanks{The authors acknowledge the support of CNRS and PAN, under the IEA project MoSART,
    and of NCBR Poland and FNR Luxembourg, under the PolLux/FNR-CORE project STV (POLLUX-VII/1/2019).}
}

\author{\IEEEauthorblockN{Jaime Arias\textsuperscript{1}, Carlos Olarte\textsuperscript{1}, Laure Petrucci\textsuperscript{1}}
  \IEEEauthorblockA{\textsuperscript{1}\textit{LIPN, CNRS UMR 7030,} \\
    \textit{Université Sorbonne Paris Nord}\\
    Villetaneuse, France \\
    \{arias, olarte, petrucci\}@lipn.univ-paris13.fr}
  \and
  \IEEEauthorblockN{{\L}ukasz Maśko\textsuperscript{2}, Wojciech Penczek\textsuperscript{2}, Teofil Sidoruk\textsuperscript{2, 3}}
  \IEEEauthorblockA{\textsuperscript{2}\textit{Institute of Computer Science, Polish Academy of Sciences} \\
    \textsuperscript{3}\textit{Faculty of Math. and Inf. Science, Warsaw University of Technology}\\
    Warsaw, Poland \\
    \{masko, penczek, t.sidoruk\}@ipipan.waw.pl}
}

\maketitle

\begin{abstract}
  Expressing attack-defence trees in a multi-agent setting allows for studying a new
  aspect of security scenarios,
  namely how the number of agents and their task assignment impact the performance,
  \eg attack time, of strategies executed by opposing coalitions.
  Optimal scheduling of agents' actions, a non-trivial problem, is thus vital.
  We discuss associated caveats and	propose an algorithm that synthesises such an assignment,
  targeting minimal attack time and using the minimal number of agents for a given
  attack-defence tree.
  We also investigate an alternative approach for the same problem using
  Rewriting Logic, starting with a simple and elegant declarative model,
  whose correctness (in terms of schedule's optimality) is self-evident.
	We then refine this specification,
	inspired by the design of our specialised algorithm,
  to obtain an efficient system that can be used as a
  playground to explore various
  aspects of attack-defence trees. We compare the two approaches
  on different benchmarks.
\end{abstract}

\begin{IEEEkeywords}
  attack-defence trees, multi-agent systems, scheduling, rewriting logic
\end{IEEEkeywords}

\section{Introduction}
\label{sec:intro}


Security of safety-critical multi-agent systems \cite {Wooldridge02a} is a major challenge. 
Attack-defence trees (\ADT/s) have been developed
to evaluate the safety of systems and to study interactions between attacker and defender parties~\cite{KordyMRS10,Zaruhi_pareto_2015}.
They provide a simple graphical formalism of possible attacker's actions to be taken in order
to attack a system and the defender's defences employed to protect the system. 
Recently, it has been proposed to model \ADT/s in the formalism of
asynchronous multi-agent systems (AMAS) extended with certain \ADT/ 
characteristics~\cite{ICECCS2019,ICFEM2020}.
In this setting, one can reason about attack/defence scenarios 
considering agent distributions over tree nodes and their impact on the feasibility and performance
(quantified by metrics such as time and cost) of attacking and defending strategies executed by specific coalitions.

\subsection{Minimal schedule with minimal number of agents}

The time metric, on which we focus here, is clearly affected by
both the number of available agents and their distribution over \ADT/ nodes. 
Hence, there arises the problem of optimal scheduling, \ie{} obtaining an assignment that achieves
the lowest possible time, while using the minimum number of agents required for an attack to be feasible. 
To that end, we first preprocess the input \ADT/, transforming it into a Directed Acyclic Graph (\DAG/),
where specific types of \ADT/ gates are replaced with sequences of nodes with normalised time 
(\ie{} duration of either zero, or the greatest common factor across all nodes of
the original \ADT/).
Because some \ADT/ constructs (namely, \gateOR gates and defences) induce multiple alternative outcomes,
we execute the scheduling algorithm itself on a number of independently considered \DAG/ variants.
For each such variant, we 
synthesise a schedule multiple times in a divide-and-conquer strategy,
adjusting the number of agents until the lowest one that produces a valid assignment is found.
Since we preserve labels during the preprocessing step, all \DAG/ nodes are traceable back to specific gates and leaves of the original \ADT/.
Thus, in the final step we ensure that the same agent is assigned to nodes of the same origin, reshuffling the schedule if necessary.


\subsection{An alternative approach: Rewriting Logic}

 We also study the optimal scheduling problem for ADTrees through the lenses
of
    Rewriting Logic (RL)  \cite{meseguer-rltcs-1992} (see also the surveys in
    \cite{meseguer-twenty-2012,DBLP:journals/jlap/DuranEEMMRT20}).
    RL is a formal model of computation whose basic building unit is a rewrite
    theory $\mathcal{R}$. 
    Roughly, the states of the
    modelled system are encoded in $\mathcal{R}$ via algebraic data types, 
    and the (non-deterministic) transitions of the system are expressed
    by a set of (conditional) rewriting rules. 
    If the theory
    $\mathcal{R}$ satisfies certain executability conditions (making  the
    mathematical and the execution semantics of $\mathcal{R}$ coincide),
    $\mathcal{R}$ can be executed in  Maude \cite{DBLP:conf/maude/2007}, a
    high-performance language and system supporting rewriting logic. 

   We start with a rewrite theory giving meaning to the different gates of
      an ADTree. The correctness of such a specification is self-evident and it
      allows us to solve the optimal scheduling problem by exploring, via a
      search procedure, all the possible paths leading to an attack.
      Unfortunately, this procedure does not scale well for complex ADTrees.
      Hence, we refine the first rewrite theory by incorporating
      some of the design principles devised in our specialised algorithm.
      We better control the non-deterministic choices in
      the specification, thus reducing the search space. The resulting theory
      can be effectively used in the case studies presented here and it opens
      the possibility of exploring different optimisation ideas and different
      aspects of ADTrees as discussed in \Cref{sec:conclusion}.

\subsection{Contributions}

In this paper, we:
\begin{enumerate*}[($i$)]
	\item present and prove the correctness of an algorithm for \ADT/s which finds an optimal assignment of the minimal number of agents for all possible \DAG/ variants of a given attack/defence scenario,
	\item show the scheduling algorithm's complexity to be quadratic in the number of nodes of its preprocessed input \DAG/, 
	\item implement the algorithm in our tool \tool, 
	\item {propose a rewrite theory, implemented in Maude,  for a general solution to the considered problem, evaluate results and compare them against those of our specialised algorithm.}
\end{enumerate*}

\subsection{Related work}
\label{sec:intro:relatedwork}

\ADT/s \cite{KordyMRS10,KMRS14}
are a popular formalism that has been implemented in a broad range
of analysis frameworks \cite{Gadyatskaya_2016,AGKS15,GIM15,ANP16},
comprehensively surveyed in \cite{KPCS13,WAFP19}.
They remain extensively studied today \cite{FilaW20}.
Of particular relevance is the \ADT/ to AMAS translation \cite{ICFEM2020},
based on the semantics from \cite{AAMASWJWPPDAM2018a}.
Furthermore, the problem discussed in this paper is clearly related to
parallel program scheduling \cite{CompSched,KwokA99}.
Due to time normalisation, it falls into the category of Unit Computational Cost (UCC) graph scheduling problems,
which can be effectively solved for tree-like structures \cite{Hu1961},
but cannot be directly applied to a set of \DAG/s.
Although a polynomial solution for interval-ordered \DAG/s was proposed by \cite{PapaYan1979},
their algorithm does not guarantee the minimal number of agents.
Due to zero-cost communication in all considered graphs, the problem can also be classified as No Communication (NC) graph scheduling.
A number of heuristic algorithms using list scheduling were proposed \cite{CompSched}, including
Highest Levels First with No Estimated Times (HLFNET),
Smallest Co-levels First with no Estimated Times (SCFNET),
and Random, where nodes in the DAG are assigned priorities randomly.
Variants assuming non-uniform node computation times are also considered, 
but are not applicable to the problem solved in this paper.
Furthermore, this class of algorithms does not aim at finding a schedule with the minimal number of processors or agents.
On the other hand, known algorithms that include such a limit, \ie for the Bounded Number of Processors (BNP) class of problems,
assume non-zero communication cost and rely on the clustering technique,
reducing communication, and thus schedule length, by mapping nodes to processing units.
Hence, these techniques are not directly applicable. 

The algorithm described in this paper can be classified as list scheduling with a fusion of HLFNET and SCFNET heuristics,
but with additional restriction on the number of agents used.
The length of a schedule is determined as the length of the critical path of a graph.
The number of minimal agents needed for the schedule is found with bisection.

Branching schedules analogous to the variants discussed in \Cref{sec:preprocess} have been previously explored, 
albeit using different models that either include probability \cite{Towsley86}
or require an additional \DAG/ to store possible executions \cite{ElRewiniA95}.
Zero duration nodes are also unique to the \ADT/ setting.

To the best of our knowledge, this is the first work dealing with agents in this context.
Rather, scheduling in multi-agent systems typically focuses on agents' \textit{choices}
in cooperative or competitive scenarios, \eg in models such as BDI \cite{NunesL14,DannT0L20}.

Rewriting logic and Maude have been extensively used for the formal analysis
and verification of systems. The reader can find in
\cite{meseguer-twenty-2012,DBLP:journals/jlap/DuranEEMMRT20} a survey of the
different techniques and applications in this field. In the context of ADTrees,
the work in \cite{DBLP:conf/crisis/KordyKB16} and the companion tool SPTool
define a rewrite theory that allows for checking the equivalence between
ADTrees featuring sequential AND gates. The work in
\cite{DBLP:conf/gramsec/EadesJB18} extends the SPTool by adding different
backend theories written in Maude: one for checking equivalence of ADTrees and
one implemented a linear-logic based semantics
\cite{DBLP:journals/fuin/HorneMT17} for it.  In none of these works and tools,
the problem of finding the optimal scheduling for agents is considered.

\subsection{Outline}

The next section briefly recalls the \ADT/ formalism.
In \Cref{sec:preprocess}, several preprocessing steps are discussed,
including transforming the input tree to a DAG, normalising node attributes,
and handling different types of nodes. 
\Cref{sec:algo} describes the main algorithm, as well as a proof of its correctness and optimality.
The algorithm, implemented in our tool \tool\cite{ADT2AMASDemoPaper},
is benchmarked  in \Cref{sec:expe}.
The rewriting logic specification is described and experimented in \Cref{sec:rewriting},
and we discuss the pros and cons with respect to the specialized algorithm proposed here. 
\Cref{sec:conclusion} concludes the paper and provides perspectives for future work. 

{This paper is an extended version of \cite{ICECCS2022}. 
    From the theoretical point of view, 
    the rewriting semantics in \Cref{sec:rewriting} is completely new. 
    From the practical side, we provide another tool, \toolM,  that enacts the
    rewriting approach. }

\section{Attack-Defence Trees}
\label{sec:ADT}


To keep the paper self-contained, we briefly recall the basics of \ADT/s
and their translation to a multi-agent setting.

\subsection{Attack-defence trees}
\label{sec:adt:adt}

\ADT/s are a well-known formalism that models security scenarios
as an interplay between attacking and defending parties.
\Cref{fig:adt:constructs} depicts the basic constructs
used throughout the paper.
For a more comprehensive overview, we refer the reader to~\cite{ICFEM2020}.


\begin{figure}[!!htb]
	\begingroup
		\def\nodesubtree{\node[isosceles triangle, isosceles triangle apex angle=70, shape border rotate=90, minimum size=6.5mm]}
		\def\nodeattack{\node[state,draw=red,ultra thick]}
		\def\nodedefence{\node[rectangle,draw=Green,ultra thick]}
	\captionsetup[subfigure]{justification=centering}
	\centering
	\begin{subfigure}[b]{0.2\linewidth}
		\centering
		\scalebox{0.5}{
			\begin{tikzpicture}
				\normalsize
				\nodeattack[minimum size=6mm] (LA) {$\actAttack{}$};
			\end{tikzpicture}
		}
		\subcaption{leaf (attack)}
	\end{subfigure}
	\hfill
	\begin{subfigure}[b]{0.2\linewidth}
		\centering
		\scalebox{0.5}{
			\begin{tikzpicture}[every node/.style={ultra thick,draw=red,minimum size=6mm}]
				\normalsize
				\node[and gate US,point up,logic gate inputs=nn] (A)
					{\rotatebox{-90}{\large$\ActAttack$}};
				\nodesubtree at (-1,-1.3) (a1) {$\!\actAttack_1\!\!$};
				\nodesubtree at ( 1,-1.3) (an) {$\!\actAttack_n\!\!$};
				\node[draw=none] at (0,-1.3) {$\cdots$};
				\draw (a1.north) -- ([yshift=1.5mm]a1.north) -| (A.input 1);
				\draw (an.north) -- ([yshift=1.3mm]an.north) -| (A.input 2);
			\end{tikzpicture}
		}
		\subcaption{\gateAND}
	\end{subfigure}
	\hfill
	\begin{subfigure}[b]{0.2\linewidth}
		\centering
		\scalebox{0.5}{
			\begin{tikzpicture}[every node/.style={ultra thick,draw=red,minimum size=6mm}]
				\normalsize
				\node[or gate US,point up,logic gate inputs=nn] (A)
					{\rotatebox{-90}{\large$\ActAttack$}};
				\nodesubtree at (-1,-1.3) (a1) {$\!\actAttack_1\!\!$};
				\nodesubtree at ( 1,-1.3) (an) {$\!\actAttack_n\!\!$};
				\node[draw=none] at (0,-1.3) {$\cdots$};
				\draw (a1.north) -- ([yshift=1.5mm]a1.north) -| (A.input 1);
				\draw (an.north) -- ([yshift=1.3mm]an.north) -| (A.input 2);
			\end{tikzpicture}
		}
		\subcaption{\gateOR}
	\end{subfigure}
	\hfill
		\begin{subfigure}[b]{0.2\linewidth}
		\centering
		\scalebox{0.5}{
			\begin{tikzpicture}[every node/.style={ultra thick,draw=red,minimum size=6mm}]
				\normalsize
				\node[and gate US,point up,logic gate inputs=nn, seq=4pt] (A)
					{\rotatebox{-90}{\large$\ActAttack$}};
				\nodesubtree at (-1,-1.3) (a1) {$\!\actAttack_1\!\!$};
				\nodesubtree at ( 1,-1.3) (an) {$\!\actAttack_n\!\!$};
				\node[draw=none] at (0,-1.3) {$\cdots$};
				\draw (a1.north) -- ([yshift=1.5mm]a1.north) -| (A.input 1);
				\draw (an.north) -- ([yshift=1.5mm]an.north) -| (A.input 2);
			\end{tikzpicture}
		}
		\subcaption{\gateSAND}
		\end{subfigure}
	\hfill
		\begin{subfigure}[b]{0.2\linewidth}
			\centering
			\scalebox{0.5}{
				\begin{tikzpicture}
					\normalsize
					\nodedefence[minimum size=6mm] (LD) {$\actDefence{}$};
				\end{tikzpicture}
			}
			\subcaption{leaf (defence)}
		\end{subfigure}
	\hfill
		\begin{subfigure}[b]{0.2\linewidth}
		\centering
		\scalebox{0.5}{
			\begin{tikzpicture}[every node/.style={ultra thick,draw=red,minimum size=6mm}]
				\normalsize
				\node[and gate US,point up,logic gate inputs=ni] (A)
					{\rotatebox{-90}{\large$\ActAttack$}};
				\nodesubtree             at (-1,-1.36) (a) {$\actAttack$};
				\nodesubtree[draw=Green] at ( 1,-1.32) (d) {$\!\actDefence\!$};
				\draw (a.north) -- ([yshift=1.4mm]a.north) -| (A.input 1);
				\draw (d.north) -- ([yshift=1.5mm]d.north) -| (A.input 2);
			\end{tikzpicture}
		}
		\subcaption{\gateCAND}
	\end{subfigure}
	\hfill
		\begin{subfigure}[b]{0.2\linewidth}
		\centering
		\scalebox{0.5}{
			\begin{tikzpicture}[every node/.style={ultra thick,draw=red,minimum size=6mm}]
				\normalsize
				\node[or gate US,point up,logic gate inputs=ni] (A)
					{\rotatebox{-90}{\large$\ActAttack$}};
				\nodesubtree             at (-1,-1.28) (a) {$\actAttack$};
				\nodesubtree[draw=Green] at ( 1,-1.24) (d) {$\!\actDefence\!$};
				\draw (a.north) -- ([yshift=1.4mm]a.north) -| (A.input 1);
				\draw (d.north) -- ([yshift=1.5mm]d.north) -| (A.input 2);
			\end{tikzpicture}
		}
		\subcaption{\gateNODEF}
	\end{subfigure}
	\hfill
		\begin{subfigure}[b]{0.2\linewidth}
		\centering
		\scalebox{0.5}{
			\begin{tikzpicture}[every node/.style={ultra thick,draw=red,minimum size=6mm}]
				\normalsize
				\node[and gate US,point up,logic gate inputs=ni, seq=9pt] (A)
					{\rotatebox{-90}{\large$\ActAttack$}};
				\nodesubtree             at (-1,-1.36) (a) {$\actAttack$};
				\nodesubtree[draw=Green] at ( 1,-1.32) (d) {$\!\actDefence\!$};
				\draw (a.north) -- ([yshift=1.4mm]a.north) -| (A.input 1);
				\draw (d.north) -- ([yshift=1.5mm]d.north) -| (A.input 2);
			\end{tikzpicture}
		}
		\subcaption{\gateSCAND}
	\end{subfigure}
	\endgroup

	\caption{Basic \ADT/ constructs \label{fig:adt:constructs}}
\end{figure}

Attacking and defending actions are depicted in red and green, respectively.
Leaves represent individual actions at the highest level of granularity.
Different types of gates allow for modelling increasingly broad intermediary goals,
all the way up to the root, which corresponds to the overall objective.
\gateOR and \gateAND gates are defined analogously to their logical counterparts.
\gateSAND is a sequential variant of the latter, \ie the entire subtree $a_i$ needs to be completed before handling $a_{i+1}$.
While only shown in attacking subtrees here, these gates may refine defending goals in the same way.
Reactive or passive countering actions can be expressed using gates
\gateCAND (counter defence; successful iff \actAttack{} succeeds and \actDefence{} fails),
\gateNODEF (no defence; successful iff either \actAttack{} succeeds or \actDefence{} fails),
and \gateSCAND (failed reactive defence; sequential variant of \gateCAND, where \actAttack{} occurs first).
We collectively refer to gates and leaves as \textit{nodes}.

\ADT/ nodes may additionally have numerical \textit{attributes},
\eg the time needed for an attack, or its financial cost.
Boolean functions over these attributes, called \textit{conditions},
may then be associated with counter-defence nodes to serve as additional constraints
for the success or failure of a defending action.

In the following, the \emph{treasure hunters} \ADT/ in \Cref{fig:adt:treasure}
will be used as a running example.
While both the gatekeeper \leaf{b} and the door \leaf{f}
need to be taken care of to steal the treasure (\gate{ST}),
just one escape route (either \leaf{h} or \leaf{e}) is needed
to flee (\gate{GA}), with \gate{TF} enforcing sequentiality.


\begin{figure}[!!htb]
	\hspace{-1em}%
	\begin{subfigure}[b]{0.45\linewidth}
		\centering
		\scalebox{.75}{
			\begin{tikzpicture}
				[every node/.style={ultra thick,draw=red,minimum size=6mm}]
				\node[and gate US,point up,logic gate inputs=ni] (ca)
				{\rotatebox{-90}{\texttt{TS}}};
				\node[rectangle,draw=Green,minimum size=8mm, below = 5mm of ca.west, xshift=10mm] (d) {\texttt{p}};
				\draw (d.north) -- ([yshift=0.28cm]d.north) -| (ca.input 2);
				\node[and gate US,point up,logic gate inputs=nn, seq=4pt, below = 9mm of ca.west, yshift=14mm] (A)
				{\rotatebox{-90}{\texttt{TF}}};
				\draw (A.east) -- ([yshift=0.15cm]A.east) -| (ca.input 1);
				\node[and gate US,point up,logic gate inputs=nn, below = 9mm of A.west, yshift=14mm] (a1)
				{\rotatebox{-90}{\texttt{ST}}};
				\draw (a1.east) -- ([yshift=0.15cm]a1.east) -| (A.input 1);
				\node[state, below = 4mm of a1.west, xshift=-5mm] (a1n) {\texttt{b}};
				\draw (a1n.north) -- ([yshift=0.15cm]a1n.north) -| (a1.input 1);
				\node[state, below=4mm of a1.west, xshift=5mm] (a11) {\texttt{f}};
				\draw (a11.north) -- ([yshift=0.15cm]a11.north) -| (a1.input 2);
				\node[or gate US,point up,logic gate inputs=nn, below = 10mm of A.west, yshift=-7mm] (a2)
				{\rotatebox{-90}{\texttt{GA}}};
				\draw (a2.east) -- ([yshift=0.15cm]a2.east) -| (A.input 2);
				\node[state, below = 4mm of a2.west, xshift=-5mm] (a21) {\texttt{h}};
				\draw (a21.north) -- ([yshift=0.15cm]a21.north) -| (a2.input 1);
				\node[state, below=4mm of a2.west, xshift=5mm] (a2n) {\texttt{e}};
				\draw (a2n.north) -- ([yshift=0.15cm]a2n.north) -| (a2.input 2);
			\end{tikzpicture}
		}
		\subcaption{\ADT/}
	\end{subfigure}
	\hspace{-.5em}
	\begin{subfigure}[b]{.45\linewidth}
		\centering
		\scalebox{.75}{\parbox{\linewidth}{%
				~%
				\begin{tabular}{r@{~}l@{~\;}r@{~\;}r}
					\multicolumn{2}{l}{\bf Name} & {\bf Cost}         & {\bf Time}              \\
					\hline
					\gate{TS}                    & (treasure stolen)  &                &        \\
					\leaf{p}                     & (police)           & \EUR{100}      & 10 min \\
					\gate{TF}                    & (thieves fleeing)  &                &        \\
					\gate{ST}                    & (steal treasure)   & 							 & 2 min  \\
					\leaf{b}                     & (bribe gatekeeper) & \EUR{500}      & 1 h    \\
					\leaf{f}                     & (force arm. door) 	& \EUR{100}      & 2 h    \\
					\gate{GA}                    & (get away)         &                &        \\
					\leaf{h}                     & (helicopter)       & \EUR{500}      & 3 min  \\
					\leaf{e}                     & (emergency exit)   &                & 10 min
				\end{tabular}
			}}
		\vspace{1ex}
		\subcaption{Attributes of nodes}
	\end{subfigure}
	\caption{Running example: treasure hunters}
	\label{fig:adt:treasure}
\end{figure}

\subsection{Translation to extended AMAS}
\label{sec:adt:adt2amas}

Asynchronous multi-agent systems (AMAS) \cite{AAMASWJWPPDAM2018a} are essentially networks of automata,
which synchronise on shared transitions and interleave private ones for asynchronous execution.
An extension of this formalism with attributes and conditional constraints to model \ADT/s,
and the translation of the latter to extended AMAS, were proposed in \cite{ICFEM2020}.
Intuitively, each node of the \ADT/ corresponds to a single automaton in the resulting network.
Specific patterns, embedding reductions to minimise state space explosion \cite{ICECCS2019},
are used for different types of \ADT/ constructs.
As the specifics exceed the scope and space of this paper, we refer the reader to \cite{AAMASWJWPPDAM2018a} for the AMAS semantics,
and to \cite{ICFEM2020} for the details on the translation.


In the multi-agent setting, groups of agents working for the attacking and defending parties can be considered.
Note that the \textit{feasibility} of an attack is not affected by the number or distribution of agents over \ADT/ nodes,
as opposed to some \textit{performance} metrics, such as time 
(\eg a lone agent can handle all the actions sequentially, albeit usually much slower).

\subsection{Assignment of agents for \ADT/s}
\label{sec:adt:sched}

Consequently, the optimal distribution of agent coalitions is of vital importance for both parties,
allowing them to prepare for multiple scenarios, depending on how many agents they can afford to recruit
(thereby delaying or speeding up the completion of the main goal).
For instance, the thieves in \Cref{fig:adt:treasure}, knowing the police response time, would have to plan accordingly
by bringing a sufficiently large team and, more importantly, schedule their tasks to make the most of these numbers.
Thus, we can formulate two relevant and non-trivial scheduling problems.
%
{\em The first one}, not directly addressed here, is obtaining the assignment using a given number of agents that results 
in optimal execution time.
{\em The second one}, on which we focus in this paper, is synthesising an assignment that achieves a particular execution 
time using the least possible number of agents.
Typically, the minimum possible time is of interest here.
As we show in \Cref{sec:preprocess}, this time can be computed from the structure of the input \ADT/ itself
(and, of course, the time attribute of nodes).
However, our approach can also target a longer attack time if desired.
In the next section, we discuss it in more detail as normalisation of the input tree is considered, 
along with several other preprocessing steps.

\section{Preprocessing the tree}
\label{sec:preprocess}

In this preprocessing step, an \ADT/ is transformed into
\DAG/s (\emph{Directed Acyclic Graphs}) of actions of the same duration.
This is achieved by splitting nodes into sequences of such actions, mimicking the
scheduling enforced by \ADT/s sequential gates, and considering the different possibilities
of defences.
Therefore, we introduce a sequential node \gateSEQ, which only waits for some input,
processes it and produces some output. It is depicted as a lozenge (see 
\Cref{fig:pre:treasure:1}).


In what follows, we assume that one
time unit is the greatest common factor of time durations across all nodes in the
input \ADT/, \ie{} $\tunit = \mathit{gcf}(t_{N_1} \dots t_{N_{|\ADT/|}})$.
By \textit{time slots}, we refer to fragments of the schedule whose length is $\tunit$.
That is, after normalisation, one agent can handle exactly one node of non-zero duration
within a single time slot.
%
Note that, during the preprocessing steps described in this section, node labels are
preserved to ensure backwards traceability. 
Their new versions are either primed or indexed.

\subsection{Nodes with no duration}
\label{sec:pre:zero}
It happens that several nodes have no time parameter set, and are thus considered to have a duration of $0$. 
Such nodes play essentially a structuring role.
Since they do not take any time, the following proposition is straightforward.

\begin{proposition}
	\label{prop:pre:zero}
	Nodes with duration $0$ can always be scheduled immediately before their parent node
	or after their last occurring child, using the same agent in the same time slot.
\end{proposition}

Preprocessing introduces nodes similar to \gateSEQ{} but with 0~duration, called \gateNULL and depicted as trapeziums (Fig.~\ref{fig:pre:treasure:sched}).

\subsection{Normalising time}
\label{sec:pre:norm}

The first preprocessing step prior to applying the scheduling algorithm
normalises the time parameter of nodes.

\begin{proposition}
	\label{prop:pre:time}
	Any node $N$ of duration $t_N = n \times \tunit, n\neq0$ can be
	replaced with an equivalent sequence consisting of a node $N'$ 
	(differing from $N$ only in its $0$ duration) and $n$ \gateSEQ nodes $N_1$,
	\ldots, $N_n$ of duration $\tunit$.
\end{proposition}

\subsection{Scheduling enforcement}
\label{sec:pre:sched}
\gateSAND nodes 
enforce some scheduling, 
and are transformed into a sequence containing their subtrees and \gateNULL nodes.

\begin{proposition}
	\label{prop:pre:sand}
	Any \gateSAND node $N$ with children subtrees $T_1$, \ldots, $T_n$ can be replaced
	with an equivalent sequence $T_1$, $N_1$, $T_2$, \ldots, $N_{n-1}$, $T_n$, $N_n$, where
	each $N_i$ is a \gateNULL node, its input is the output of $T_i$ and its outputs
	are the leaves of $T_{i+1}$ (except for $N_n$ which has the same output as $N$ if
	any).
\end{proposition}

\subsection{Handling defences}
\label{sec:pre:def}
The scheduling we are seeking to obtain will guarantee that the necessary attacks are performed.
Hence, when dealing with defence nodes, we can assume that all attacks are successful.
However, they may not be mandatory, in which case they should be avoided so as to
obtain a better scheduling of agents.

Taking into account each possible choice of defences will lead to as many \DAG/s representing
the attacks to be performed. This allows for answering the question: ``What is the
minimal schedule of attackers if these defences are operating?''



\emph{Composite defences.} Defences resulting from an \gateAND{}, \gateSAND
or \gateOR
between several defences are operating according to the success of their subtrees:
for \gateAND{} and \gateSAND{}, all subtrees should be operating, while only one is
necessary for \gateOR{}. This can easily be computed by a boolean bottom-up labelling
of nodes. Note that different choices of elementary defences can lead to disabling
the same higher-level composite defence, thus limiting the number of \DAG/s that will
need to be considered.

\emph{No Defence nodes (\gateNODEF).} 
A \gateNODEF succeeds if its attack succeeds or its defence fails. 
Hence, if the defence is not operating, the attack is not necessary.
Thus, the \gateNODEF{} node can be replaced by a \gateNULL{} node without children,
and the children subtrees 
deleted.
On the contrary, if the defence is operating, the attack must take place. The defence
subtree is deleted, while the attack one is kept, and the \gateNODEF{} node can be
replaced by a \gateNULL{} node, as 
depicted in \Cref{fig:pre:nodef}.

\emph{Counter Defence (\gateCAND{}) and Failed Reactive Defence (\gateSCAND
	{}) nodes.} A \gateCAND{} succeeds if its attack is successful and its defence is
not. A \gateSCAND{} additionally specifies that the defence
takes place after the attack. In both cases, if the defence is not operating, its
subtree is deleted, while the attack one is kept, and the \gateCAND{} (or \gateSCAND{})
node can be replaced by a \gateNULL{} node, as in \Cref{fig:pre:nodef:dok}.
Otherwise, the \gateCAND{} (or \gateSCAND{}) node is deleted, as well as its subtrees.
Moreover, it transmits its failure recursively to its parents, until a choice of another
branch is possible. Thus, all ancestors are deleted bottom up until an \gateOR is
reached.


Thus, we have a set of \DAG/s with attack nodes only.

\begin{figure}[!!htb]
\centering
\begingroup
\def\nodesubtree{\node[isosceles triangle, isosceles triangle apex angle=70, shape border rotate=90, minimum size=6.5mm]}

\begin{subfigure}[b]{0.3\linewidth}
	\centering
	\scalebox{0.75}{
		\begin{tikzpicture}
			[every node/.style={ultra thick,draw=red,minimum size=6mm},
				node distance=1cm]
			\normalsize
			\node[or gate US,point up,logic gate inputs=ni] (A)
			{\rotatebox{-90}{$A$}};
			\nodesubtree             at (-1,-1.28) (a) {$a$};
			\nodesubtree[draw=Green] at ( 1,-1.24) (d) {$d$};
			\draw (a.north) -- ([yshift=1.4mm]a.north) -| (A.input 1);
			\draw (d.north) -- ([yshift=0.6mm]d.north) -| (A.input 2);
		\end{tikzpicture}
	}
	\subcaption{\gateNODEF{} node\label{fig:pre:nodef:gate}}
\end{subfigure}
\begin{subfigure}[b]{0.3\linewidth}
	\centering
	\scalebox{0.75}{
		\begin{tikzpicture}
			[every node/.style={ultra thick,draw=red,minimum size=6mm},
				node distance=1cm]
			\normalsize
			\node[draw=none] at (0,0) () {};
			\node[gateNULL] at (0,1.28) (A') {$A'$};
		\end{tikzpicture}
	}
	\subcaption{Case $d$ fails\label{fig:pre:nodef:dnok}}
\end{subfigure}
\begin{subfigure}[b]{0.3\linewidth}
	\centering
	\scalebox{0.75}{
		\begin{tikzpicture}
			[every node/.style={ultra thick,draw=red,minimum size=6mm},
				node distance=1cm]
			\normalsize
			\node[gateNULL] (A') {$A'$};
			\nodesubtree             at (0,-1.28) (a) {$a$};
			\draw (a.north) -- (A');
		\end{tikzpicture}
	}
	\subcaption{Case $d$ operates\label{fig:pre:nodef:dok}}
\end{subfigure}
\endgroup
\caption{Handling \gateNODEF $A$\label{fig:pre:nodef}}
\end{figure}
\subsection{Handling OR branches}
\label{sec:pre:or}
\gateOR nodes give the choice between several series of actions,
only one of which will be chosen in an optimal assignment of events.
However, one cannot simply keep the shortest branch of an \gateOR node and prune all others.
Doing so minimises attack time, but not necessarily the number of agents.
In particular, a slightly longer, but narrower branch may require fewer agents without increasing attack time,
provided there is a longer sequence elsewhere in the \DAG/.
Consequently, only branches that are guaranteed not to lead to an optimal assignment can be pruned,
which is the case when a branch is the longest one in the entire graph.
All other cases need to be investigated, leading to multiple variants depending on
the \gateOR branch executed, similar to the approach for defence nodes.

\subsection{Preprocessing the treasure hunters \ADT/}
\label{sec:pre:ex}

\Cref{fig:pre:treasure:1,fig:pre:treasure:2} detail the preprocessing of the treasure
hunters example step by step. The time unit is one minute. Long sequences of \gateSEQ
{} are shortened with dotted lines.
Note that when handling the defence, at step 3, we should obtain two \DAG/s corresponding
to the case where the defence fails (see \Cref{fig:pre:treasure:def}), or where the
defence is successful. This latter case leads to an empty \DAG/ where no attack can
succeed. Therefore, we can immediately conclude that if the police is successful,
there is no scheduling of agents.

\begin{figure}[!!htb]
	\captionsetup[subfigure]{justification=centering}
	\centering
	
	\begin{subfigure}[b]{\linewidth}
		\centering
		\scalebox{0.5}{
		\begin{tikzpicture}[every node/.style={ultra thick,draw=red,minimum size=6mm},
				node distance=1.5cm]
			\node[and gate US,point up,logic gate inputs=ni] (ts)
			{\rotatebox{-90}{\gate{TS'}}};
			\node[gateSEQ,draw=Green,minimum size=8mm,
				below = 5mm of ts.west, xshift=22mm]
			(p10) {\leaf{p_{10}}};
			\node[gateSEQ,draw=Green,minimum size=8mm,below of= p10]
			(p1) {\leaf{p_{1}}};
			\node[rectangle,draw=Green,minimum size=8mm,below of= p1]
			(p) {\leaf{p'}};
			\draw (p10.north) -- ([yshift=0.28cm]p10.north) -| (ts.input 2);
			\draw[dotted] (p10) --(p1);
			\draw (p) -- (p1);
			\node[and gate US,point up,logic gate inputs=nn, seq=4pt,
				below = 9mm of ts.west, yshift=28mm] (tf)
			{\rotatebox{-90}{\gate{TF'}}};
			\draw (tf.east) -- ([yshift=0.15cm]tf.east) -| (ts.input 1);
			\node[gateSEQ,minimum size=8mm,
				left of = p1, node distance=6.7cm]
			(st2) {\leaf{ST_{2}}};
			\node[gateSEQ,minimum size=8mm,below of= st2]
			(st1) {\leaf{ST_{1}}};
			\node[and gate US,point up,logic gate inputs=nn,
				left of= st1]
			(st) {\rotatebox{-90}{\gate{ST'}}};
			\draw (st2.north) -- ([yshift=0.15cm]st2.north) -| (tf.input 1);
			\draw (st2) -- (st1) -- (st.east);
			\node[gateSEQ, below = 4mm of st.west, xshift=-1.2cm]
			(b60) {\leaf{b_{60}}};
			\node[gateSEQ,minimum size=8mm,below of= b60]
			(b1) {\leaf{b_{1}}};
			\node[state,minimum size=8mm,below of= b1]
			(b) {\leaf{b'}};
			\draw (b60.north) -- ([yshift=0.15cm]b60.north) -| (st.input 1);
			\draw[dotted] (b60) --(b1);
			\draw (b) -- (b1);
			\node[gateSEQ, below = 4mm of st.west, xshift=1.2cm]
			(f120) {\leaf{f_{120}}};
			\node[gateSEQ,minimum size=8mm,below of= f120]
			(f1) {\leaf{f_{1}}};
			\node[state,minimum size=8mm,below of= f1]
			(f) {\leaf{f'}};
			\draw (f120.north) -- ([yshift=0.15cm]f120.north) -| (st.input 2);
			\draw[dotted] (f120) --(f1);
			\draw (f) -- (f1);
			\node[or gate US,point up,logic gate inputs=nn,
				above of = p1, node distance=2.5cm]
			(ga) {\rotatebox{-90}{\gate{GA'}}};
			\draw (ga.east) -- ([yshift=0.15cm]ga.east) -| (tf.input 2);
			\node[gateSEQ, below = 4mm of ga.west, xshift=-1.2cm]
			(h3) {\leaf{h_{3}}};
			\node[gateSEQ,minimum size=8mm,below of= h3]
			(h2) {\leaf{h_{2}}};
			\node[gateSEQ,minimum size=8mm,below of= h2]
			(h1) {\leaf{h_{1}}};
			\node[state,minimum size=8mm,below of= h1]
			(h) {\leaf{h'}};
			\draw (h3.north) -- ([yshift=0.15cm]h3.north) -| (ga.input 1);
			\draw (h3) -- (h2) -- (h1) -- (h);
			\node[gateSEQ, below = 4mm of ga.west, xshift=1.2cm]
			(e10) {\leaf{e_{10}}};
			\node[gateSEQ,minimum size=8mm,below of= e10]
			(e1) {\leaf{e_{1}}};
			\node[state,minimum size=8mm,below of= e1]
			(e) {\leaf{e'}};
			\draw (e10.north) -- ([yshift=0.15cm]e10.north) -| (ga.input 2);
			\draw[dotted] (e10) --(e1);
			\draw (e) -- (e1);
		\end{tikzpicture}
	}
		\subcaption{time normalisation}
		\label{fig:pre:treasure:1}
	\end{subfigure}
	
	\begin{subfigure}[b]{0.35\linewidth}
		\centering
		\scalebox{0.5}{
			\begin{tikzpicture}[every node/.style={ultra thick,draw=red,minimum size=6mm},
					node distance=1.5cm]
				\node[and gate US,point up,logic gate inputs=ni] (ts)
				{\rotatebox{-90}{\gate{TS'}}};
				\node[gateSEQ,draw=Green,minimum size=8mm,
					below = 5mm of ts.west, xshift=10mm]
				(p10) {\leaf{p_{10}}};
				\node[gateSEQ,draw=Green,minimum size=8mm,below of= p10]
				(p1) {\leaf{p_{1}}};
				\node[rectangle,draw=Green,minimum size=8mm,below of= p1]
				(p) {\leaf{p'}};
				\draw (p10.north) -- ([yshift=0.28cm]p10.north) -| (ts.input 2);
				\draw[dotted] (p10) --(p1);
				\draw (p) -- (p1);
				\node[gateNULL,minimum size=8mm,yshift=1.7mm,
					left of= p10, node distance=2.5cm] (tf2)
				{\gate{TF'_2}};
				\draw (tf2.north) -- ([yshift=0.25cm]tf2.north) -| (ts.input 1);
				\node[or gate US,point up,logic gate inputs=nn,
					above of = p1, node distance=2.5cm]
				(ga) {\rotatebox{-90}{\gate{GA'}}};
				\draw (ga.east) -- ([yshift=0.15cm]ga.east) -| (tf2);
				\node[gateSEQ, below = 4mm of ga.west, xshift=-1.2cm]
				(h3) {\leaf{h_{3}}};
				\node[gateSEQ,minimum size=8mm,below of= h3]
				(h2) {\leaf{h_{2}}};
				\node[gateSEQ,minimum size=8mm,below of= h2]
				(h1) {\leaf{h_{1}}};
				\node[state,minimum size=8mm,below of= h1]
				(h) {\leaf{h'}};
				\draw (h3.north) -- ([yshift=0.15cm]h3.north) -| (ga.input 1);
				\draw (h3) -- (h2) -- (h1) -- (h);
				\node[gateSEQ, below = 4mm of ga.west, xshift=1.2cm]
				(e10) {\leaf{e_{10}}};
				\node[gateSEQ,minimum size=8mm,below of= e10]
				(e1) {\leaf{e_{1}}};
				\node[state,minimum size=8mm,below of= e1]
				(e) {\leaf{e'}};
				\draw (e10.north) -- ([yshift=0.15cm]e10.north) -| (ga.input 2);
				\draw[dotted] (e10) --(e1);
				\draw (e) -- (e1);
				\node[gateNULL,minimum size=8mm,xshift=1.2cm,
					below of= h]
				(tf1) {\gate{TF'_1}};
				\draw (tf1) -- (h);
				\draw (tf1) -- (e);
				\node[gateSEQ,minimum size=8mm,below of = tf1]
				(st2) {\leaf{ST_{2}}};
				\node[gateSEQ,minimum size=8mm,below of= st2]
				(st1) {\leaf{ST_{1}}};
				\node[and gate US,point up,logic gate inputs=nn,
					left of= st1]
				(st) {\rotatebox{-90}{\gate{ST'}}};
				\draw (st2) -- (tf1);
				\draw (st2) -- (st1) -- (st.east);
				\node[gateSEQ, below = 4mm of st.west, xshift=-1.2cm]
				(b60) {\leaf{b_{60}}};
				\node[gateSEQ,minimum size=8mm,below of= b60]
				(b1) {\leaf{b_{1}}};
				\node[state,minimum size=8mm,below of= b1]
				(b) {\leaf{b'}};
				\draw (b60.north) -- ([yshift=0.15cm]b60.north) -| (st.input 1);
				\draw[dotted] (b60) --(b1);
				\draw (b) -- (b1);
				\node[gateSEQ, below = 4mm of st.west, xshift=1.2cm]
				(f120) {\leaf{f_{120}}};
				\node[gateSEQ,minimum size=8mm,below of= f120]
				(f1) {\leaf{f_{1}}};
				\node[state,minimum size=8mm,below of= f1]
				(f) {\leaf{f'}};
				\draw (f120.north) -- ([yshift=0.15cm]f120.north) -| (st.input 2);
				\draw[dotted] (f120) --(f1);
				\draw (f) -- (f1);
			\end{tikzpicture}
		}
		\subcaption{Scheduling\\enforcement}
		\label{fig:pre:treasure:sched}
	\end{subfigure}
	\hfill
	\begin{subfigure}[b]{0.25\linewidth}
		\centering
		\scalebox{0.5}{
			\begin{tikzpicture}[every node/.style={ultra thick,draw=red,minimum size=6mm},
					node distance=1.5cm]
				\node[gateNULL,minimum size=8mm]
				(ts) {\gate{TS'}};
				\node[gateNULL,minimum size=8mm,below of= ts]
				(tf2) {\gate{TF'_2}};
				\draw (tf2) -- (ts);
				\node[or gate US,point up,logic gate inputs=nn,
					left of = tf2]
				(ga) {\rotatebox{-90}{\gate{GA'}}};
				\draw (ga.east) -- ([yshift=0.15cm]ga.east) -| (tf2);
				\node[gateSEQ, below = 4mm of ga.west, xshift=-1.2cm]
				(h3) {\leaf{h_{3}}};
				\node[gateSEQ,minimum size=8mm,below of= h3]
				(h2) {\leaf{h_{2}}};
				\node[gateSEQ,minimum size=8mm,below of= h2]
				(h1) {\leaf{h_{1}}};
				\node[state,minimum size=8mm,below of= h1]
				(h) {\leaf{h'}};
				\draw (h3.north) -- ([yshift=0.15cm]h3.north) -| (ga.input 1);
				\draw (h3) -- (h2) -- (h1) -- (h);
				\node[gateSEQ, below = 4mm of ga.west, xshift=1.2cm]
				(e10) {\leaf{e_{10}}};
				\node[gateSEQ,minimum size=8mm,below of= e10]
				(e1) {\leaf{e_{1}}};
				\node[state,minimum size=8mm,below of= e1]
				(e) {\leaf{e'}};
				\draw (e10.north) -- ([yshift=0.15cm]e10.north) -| (ga.input 2);
				\draw[dotted] (e10) --(e1);
				\draw (e) -- (e1);
				\node[gateNULL,minimum size=8mm,xshift=1.2cm,
					below of= h]
				(tf1) {\gate{TF'_1}};
				\draw (tf1) -- (h);
				\draw (tf1) -- (e);
				\node[gateSEQ,minimum size=8mm,below of = tf1]
				(st2) {\leaf{ST_{2}}};
				\node[gateSEQ,minimum size=8mm,below of= st2]
				(st1) {\leaf{ST_{1}}};
				\node[and gate US,point up,logic gate inputs=nn,
					left of= st1]
				(st) {\rotatebox{-90}{\gate{ST'}}};
				\draw (st2) -- (tf1);
				\draw (st2) -- (st1) -- (st.east);
				\node[gateSEQ, below = 4mm of st.west, xshift=-1.2cm]
				(b60) {\leaf{b_{60}}};
				\node[gateSEQ,minimum size=8mm,below of= b60]
				(b1) {\leaf{b_{1}}};
				\node[state,minimum size=8mm,below of= b1]
				(b) {\leaf{b'}};
				\draw (b60.north) -- ([yshift=0.15cm]b60.north) -| (st.input 1);
				\draw[dotted] (b60) --(b1);
				\draw (b) -- (b1);
				\node[gateSEQ, below = 4mm of st.west, xshift=1.2cm]
				(f120) {\leaf{f_{120}}};
				\node[gateSEQ,minimum size=8mm,below of= f120]
				(f1) {\leaf{f_{1}}};
				\node[state,minimum size=8mm,below of= f1]
				(f) {\leaf{f'}};
				\draw (f120.north) -- ([yshift=0.15cm]f120.north) -| (st.input 2);
				\draw[dotted] (f120) --(f1);
				\draw (f) -- (f1);
			\end{tikzpicture}
		}
		\subcaption{Handling\\failed defence}
		\label{fig:pre:treasure:def}
	\end{subfigure}
	\hfill
		\begin{subfigure}[b]{0.35\linewidth}
		\centering
		\scalebox{0.5}{
			\begin{tikzpicture}[every node/.style={ultra thick,
							draw=red,minimum size=6mm},
					node distance=1.5cm]
				\node[gateNULL,minimum size=8mm,
					depth label=125,
					label={[draw=none,green!60!black]120:depth},
					level label=0,
					label={[draw=none,blue]60:level}]
				(ts) {\gate{TS'}};
				\node[gateNULL,minimum size=8mm,below of= ts,
					depth label=125,level label=0]
				(tf2) {\gate{TF'_2}};
				\draw (tf2) -- (ts);
				\node[or gate US,point up,logic gate inputs=nn,
					left of = tf2]
				(ga) {\rotatebox{-90}{\gate{GA'}}};
				\node[draw=none,green!60!black,xshift=-0.7cm]
				at (ga.center) (){125};
				\node[draw=none,blue,xshift=0.6cm]
				at (ga.center) (){0};
				\draw (ga.east) -- ([yshift=0.15cm]ga.east) -| (tf2);
				\node[gateSEQ, below = of ga.center,yshift=0.5cm,
					depth label=125,level label=0]
				(h3) {\leaf{h_{3}}};
				\node[gateSEQ,minimum size=8mm,below of= h3,
					depth label=124,level label=1]
				(h2) {\leaf{h_{2}}};
				\node[gateSEQ,minimum size=8mm,below of= h2,
					depth label=123,level label=2]
				(h1) {\leaf{h_{1}}};
				\node[state,minimum size=8mm,below of= h1,
					depth label=122,level label=3]
				(h) {\leaf{h'}};
				\draw (h3.north) -- (ga);
				\draw (h3) -- (h2) -- (h1) -- (h);
				\node[gateNULL,minimum size=8mm,below of= h,
					depth label=122,level label=3]
				(tf1) {\gate{TF'_1}};
				\draw (tf1) -- (h);
				\node[gateSEQ,minimum size=8mm,below of = tf1,
					depth label=122,level label=3]
				(st2) {\leaf{ST_{2}}};
				\node[gateSEQ,minimum size=8mm,below of= st2,
					depth label=121,level label=4]
				(st1) {\leaf{ST_{1}}};
				\node[and gate US,point up,logic gate inputs=nn,
					left of= st1]
				(st) {\rotatebox{-90}{\gate{ST'}}};
				\node[draw=none,green!60!black,xshift=-0.8cm]
				at (st.center) (){120};
				\node[draw=none,blue,xshift=0.6cm]
				at (st.center) (){5};
				\draw (st2) -- (tf1);
				\draw (st2) -- (st1) -- (st.east);
				\node[gateSEQ, below = 4mm of st.west,xshift=-1.2cm,
					depth label=60,level label=5]
				(b60) {\leaf{b_{60}}};
				\node[gateSEQ,minimum size=8mm,below of= b60,
					depth label=1,level label=64]
				(b1) {\leaf{b_{1}}};
				\node[state,minimum size=8mm,below of= b1,
					depth label=0,level label=65]
				(b) {\leaf{b'}};
				\draw (b60.north) -- ([yshift=0.15cm]b60.north) -|
				(st.input 1);
				\draw[dotted] (b60) --(b1);
				\draw (b) -- (b1);
				\node[gateSEQ, below = 4mm of st.west,xshift=1.2cm,
					depth label=120,level label=5]
				(f120) {\leaf{f_{120}}};
				\node[gateSEQ,minimum size=8mm,below of= f120,
					depth label=1,level label=124]
				(f1) {\leaf{f_{1}}};
				\node[state,minimum size=8mm,below of= f1,
					depth label=0,level label=125]
				(f) {\leaf{f'}};
				\draw (f120.north) -- ([yshift=0.15cm]f120.north) -|
				(st.input 2);
				\draw[dotted] (f120) --(f1);
				\draw (f) -- (f1);
			\end{tikzpicture}
		}
		\subcaption{Handling \gateOR node, computing depth/level}%
		\label{fig:sched:pruneNlabel}
	\end{subfigure}

	\caption{Treasure hunters \ADT/: preprocessing steps (top, left, middle) and initial part of the main algorithm (bottom right)\label{fig:pre:treasure:2}}
\end{figure}

\section{Best minimal agent assignment}
\label{sec:algo}






At this stage, we have \DAG/s where nodes are either (i) a leaf, or of type \gateAND
{}, \gateOR
{}, or \gateNULL, all with duration $0$ or (ii) of type \gateSEQ{} with duration
$\tunit{}$. Their branches mimic the possible runs in the system.

The algorithm's input is a set of \DAG/s preprocessed as described in \Cref{sec:preprocess},
corresponding to possible configurations of defence nodes' outcomes and choices of \gateOR branches in the original \ADT/.
%
For each of these \DAG/s, $n$ denotes the number of \gateSEQ{} nodes (all other ones have 0-duration).
Furthermore, nodes (denoted by \N) have some attributes: their $\type$; four integers
$\depth$, $\level$, $\agent$ and $\slot$, initially with value 0.
The values of $\depth$ and $\level$ denote, respectively,
the height of a node's tallest subtree and the distance from the root
(both without counting the zero duration nodes). The attributes
$\agent$ and $\slot$ store the node's assignment in the schedule.

\subsection{Depth and level of nodes}
\label{sec:algo:depthlevel}

We first compute the nodes' depth and level, handled by procedures $\textsc{DepthNode}$ and $\textsc{LevelNode}$, respectively.
They explore the \DAG/ in a DFS (\emph{depth first search}) manner, starting from the root.
Both attributes are assigned recursively, with $\depth$ computed during backtracking, \ie starting from the leaves.
There are slight differences in the way specific node types are handled; we refer the reader to~\cite{ICECCS2022} for the details.

\subsection{Number of agents: upper and lower bounds}
\label{sec:algo:bounds}

The upper bound on the number of agents is obtained from the maximal width of the preprocessed \DAG/,
\ie the maximal number of \gateSEQ nodes assigned the same value of \textit{level}.
These nodes must be executed in parallel to guarantee that the attack is achieved in the minimal time.


The minimal attack time is obtained from the number of levels $l$ in the preprocessed \DAG/.
Note that the longest path from the root to a leaf has exactly $l$ nodes of non-zero duration.
Clearly, none of these nodes can be executed in parallel,
therefore the number of time slots cannot be smaller than $l$.
%
Thus, if an optimal schedule of $l\times\tunit{}$ is realisable,
the $n$ nodes must fit in a schedule containing $l$ time slots.
Hence, the lower bound on the number of agents is  $\lceil\frac{n}{l}\rceil$.
There is, however, no guarantee that it can be achieved, and introducing additional agents may be necessary
depending on the \DAG/ structure, \eg if there are many parallel leaves.

\subsection{Minimal schedule}
\label{sec:algo:algo}

%

The algorithm for obtaining a schedule with the minimal attack time and also minimising
the number of agents is given in Alg.~\ref{algo:minsched}.
Input \DAG/s are processed sequentially and a schedule is computed for each one.
Not restricting the output to the overall minimum allows to avoid ``no attack'' scenarios
where the time is 0 (\eg following a defence failure on a root \gateNODEF node).
Furthermore, with information on the distribution of agents for a successful minimal time attack in all cases of defences,
the defender is able to decide which defences to enable according to these results.

\begin{algorithm}[t]
	\caption{\textsc{MinSchedule}($\DAGset$)}
	\label{algo:minsched}
	$\varOutput = \emptyset$

	\While{$\DAGset \neq \emptyset$}{
		Pick $\inputDAG \in \DAGset$

		\lIf(\Comment*[f]{Skip empty \DAG/s}){$\inputDAG.n=0$}{
		\textbf{continue}}
		$\textsc{DepthNode}(root(\inputDAG))$ \Comment*{Compute depth of nodes}
		$\inputDAG \gets \inputDAG\setminus\{\N \ |\ \neg
			\N.\keep\}$\label{line:minsched:delete}

		$\textsc{LevelNode}(root(\inputDAG),0)$ \Comment*{Compute level of nodes}

		$\varSlots \gets root(\inputDAG).\depth$\label{line:minsched:numslots}

	  $\varBound \gets \lceil\frac{\inputDAG.n}{\varSlots}\rceil - 1$

		$\varMaxAgents \gets \max_j(|\{\N: \N.\type=\gateSEQ \land \N.\level=j\}|)$ \Comment*{Max. level width (concur. \gateSEQ nodes)}

		$\varUpperBound \gets \varMaxAgents$

		$\varCurrentOutput = \emptyset$

		\While{$(\varUpperBound-\varBound>1)$\label{line:minsched:newloop}}{

			$\varNumAgents \gets \varBound + \lfloor\frac{\varUpperBound-\varBound}{2}\rfloor$\label{line:minsched:adjustagents}

			$(\varCandidate,\varNodesLeft) \gets \textsc{Schedule}(\inputDAG,\varSlots,\varNumAgents)$\label{line:minsched:candidate}

		  \If(\Comment*[f]{Candidate schedule OK}){$\varNodesLeft=0$\label{line:minsched:discard}}{

				$\varUpperBound \gets \varNumAgents$

				$\varCurrentOutput \gets \varCandidate$}

			\lElse(\Comment*[f]{Cand. schedule not OK}){$\varBound=\varNumAgents$}

		} 

		\If{$\varUpperBound = \varMaxAgents$}{
			$ (\varCurrentOutput, \_) \gets \textsc{Schedule}(\inputDAG,\varSlots,\varMaxAgents)$
		}\label{line:minsched:minagents}

		$\textsc{ZeroAssign}(\inputDAG)$\label{line:minsched:zero}

		$\varOutput \gets \varOutput \cup \varCurrentOutput$

		$\DAGset \gets \DAGset \setminus \inputDAG$
	}
	\Return $\varOutput$
\end{algorithm}


The actual computation of the schedule is handled by the function \textsc{Schedule} (Alg.~\ref{algo:schedCandidate}).
Starting from the root and going top-down, all \gateSEQ nodes at the current level
are added to set $\workingSet$. The other nodes at that level have a null duration
and can be scheduled afterwards with either a parent or child.
An additional check in l.~\ref{line:sched:discard} ensures that non-optimal variants
(whose longest branch exceeds a previously encountered minimum) are discarded without needlesly computing the schedule.
Nodes in $\workingSet{}$ are assigned an agent and time slot, prioritising those with
higher $\depth$ (\ie taller subtrees), as long as an agent is available.
Assigned nodes are removed from $\workingSet$, and any that remain (\eg when the bound was exceeded)
are carried over to the next level iteration.
At this point, it is possible for a parent and a child node to be in $\workingSet$ concurrently.
However, since higher $\depth$ takes precedence, they will never be scheduled in the wrong order,
and an extra check in the while loop avoids scheduling both nodes to be executed in parallel.

\begin{algorithm}[t]
	\caption{\textsc{Schedule}($\inputDAG,\varSlots,\varNumAgents$)}\label{algo:schedCandidate}

		$l \gets 0, \varSlot \gets \varSlots,\workingSet \gets \emptyset, \varNodesLeft \gets \inputDAG.n$	\label{line:sched:numslots}

		\While{$\varNodesLeft>0$ \textbf{and} $\varSlot > 0$\label{line:minsched:mainloop}}{
			$\varAgent \gets 1$



			$\workingSet \gets \workingSet \cup \{\N \ |\ \N.\type=\gateSEQ \land \N.\level=l\}$

			\If{$\exists_{\N \in \workingSet}$, \textbf{s.t.} $\N.\depth < \varSlots
			- \varSlot$\label{line:sched:discard}}{
                \Return $\emptyset, \varNodesLeft$
			}

			\While{$\varAgent \leq \varNumAgents$ \textbf{and} $\workingSet \neq \emptyset$ \textbf{and} \qquad\qquad\qquad\qquad
					(Pick $\N\in\workingSet$, \textbf{s.t.}
					$\forall_{\N' \in \workingSet} \N.\depth \geq \N'.\depth\, \land$ \qquad
					$\forall_{\N': \N'.\varSlot = \varSlot} \N' \notin ancestors(\N)) \neq \emptyset$\label{line:minsched:prioritise}}{

				$\N.\agent \gets \varAgent$

				$\N.\slot \gets \varSlot$

				$\varAgent \gets \varAgent+1,\varNodesLeft	\gets \varNodesLeft-1\label{line:minsched:parexec}$

				$\workingSet \gets \workingSet \setminus \{\N\}$
			}

			$\textsc{ReshuffleSlot}(\slot, \varAgent-1)$
			\label{line:minsched:shuffle}

			$l \gets l+1, \varSlot \gets \varSlot-1$
			}

	$output \gets \bigcup_{N\in\inputDAG} \{ (\N.\agent,\N.\slot) \}$

	\Return $output, \varNodesLeft$
\end{algorithm}


\Cref{algo:schedCandidate} calls function $\textsc{ReshuffleSlot}$ after the complete
assignment of a time slot at l.~\ref{line:minsched:shuffle}
to ensure consistent assignment of sub-actions of the same \ADT/ node.
Note that depending on $depth$, a sub-action may be moved to the next slot,
creating an interrupted schedule where an agent stops an action for one or more time units to handle another.
Alternatively, agents may collaborate, each handling a node's action for a part of its total duration.
Such assignments could be deemed unsuitable for specific scenarios where extra conditions need to be satisfied. 
In those cases,  manual reshuffling or adding extra agent(s)  is left to the user's discretion.

At this point, either the upper or the lower bound on the number of agents is adjusted,
depending on whether the resulting schedule is valid (that is, there are no nodes left to assign at the end).
Scheduling is then repeated for these updated values until the minimal number of agents is found (\ie the two bounds are equal).

After the complete computation for a given \DAG/, l.~\ref{line:minsched:zero}
calls function $\textsc{ZeroAssign}$ in order to obtain assignments
for all remaining nodes,
\ie those of zero duration. Functions $\textsc{ReshuffleSlot}$ and $\textsc{ZeroAssign}$
are detailed in \Cref{sec:algo:reshuffle,sec:algo:zero}, respectively.


Although this algorithm assumes the minimal time is of interest,
it can be easily modified to increase the number of time slots,
thus synthesising the minimal number of agents required for a successful attack of any given duration.

\subsection{Uniform assignment for SEQ nodes}
\label{sec:algo:reshuffle}

A separate subprocedure, given in \Cref{algo:reshuffle}, swaps assigned agents between nodes at the same level
so that the same agent handles all \gateSEQ nodes in sequences obtained during the time normalisation step
(\ie corresponding to a single node in the original \ADT/).

\begin{algorithm}
	\caption{\textsc{ReshuffleSlot}($slot, num\_agents$)}
	\label{algo:reshuffle}
	\For{$agent \in \{1..num\_agents\}$}{
		$\varCurrentNode \gets \N$, \textbf{s.t.} $\N.\varAgent
		= \varAgent\land \N.\slot=\slot$

		$\varParentAgent \gets \parent(\varCurrentNode).\varAgent$

		\If{$\varParentAgent \neq \varAgent\land\varParentAgent\neq 0$\label{line:reshuffle:seqonly}}{
			\If{$\exists \N' \neq \varCurrentNode$, \textbf{s.t.}
				$\quad\N'.\agent = \varParentAgent\land \N'.\slot=\slot$}{
				$\N'.\agent \gets \varAgent$ \Comment*{Swap with $\N'$ if it exists}

				$\N'.\slot \gets \slot$
			}
			$\varCurrentNode.\agent \gets \varParentAgent$

			$\varCurrentNode.\slot \gets \slot$
		}
	}
\end{algorithm}

\begin{proposition}
	\label{prop:algo:reshuffle}
	Reshuffling the assignment by swapping the agents assigned to a pair of nodes in the same slot
	does not affect the correctness of the scheduling.
\end{proposition}
\begin{proof}
See~\cite[Proposition~4]{ICECCS2022}.
\end{proof}

\subsection{Assigning nodes without duration}
\label{sec:algo:zero}


After all non-zero duration nodes have been assigned and possibly reshuffled at each
level,
Alg.~\ref{algo:zero} handles the remaining nodes.
\begin{algorithm}
	\caption{\textsc{ZeroAssign}($\inputDAG$)}
	\label{algo:zero}
		$\workingSet{}\gets\{\N\ |\ \N.\agent=0\}$ \Comment*
		{Nodes not assigned yet}
		\For{$\node{}\in\workingSet{}$\label{line:zero:seq:start}}{
			\If{$\N\in\parent(\node{})\land \N.\type=\gateSEQ{}$}{
				$\node.\agent\gets \N.\agent$

				$\node.\slot\gets \N.\slot$

				$\workingSet{}\gets\workingSet{}\setminus\{\node\}$
			}
		\label{line:zero:seq:end}}
		\While{$\workingSet{}\neq\emptyset$}{
			\For{$\node{}\in\workingSet{}$ \textbf{s.t.} $node.type \in \{\gateNULL,
			\gateOR, \gateLEAF\}$\label{line:zero:nullorleaf:start}}{
				\If{$\N.\agent\neq 0$ \textbf{s.t.} $\N\in\children(\node{})$}{
					$\node.\agent\gets \N.\agent$

					$\node.\slot\gets \N.\slot$

					$\workingSet{}\gets\workingSet{}\setminus\{\node\}$
				}
				\If{$(\children(\node)=\emptyset$\\
						$\lor (\N.\depth=0$ \textbf{s.t.}
						$\N\in\children(\node)))$}{
					$\varParentNode\gets\N\in\parent(\node)$ \textbf{s.t.} \mbox{\quad}$\forall_
					{\N'\in \parent(\node{})} \N.\slot\leq\N'.\slot$

					\If{$\varParentNode.agent\neq 0$}{
						$\node.\agent\gets \varParentNode.\agent$

						$\node.\slot\gets \varParentNode.\slot$

						$\workingSet{}\gets\workingSet{}\setminus\{\node\}$
					}
				}
			\label{line:zero:nullorleaf:end}}
			\For{$\node{}\in\workingSet{}$ \textbf{s.t.} $\node{}.\type=\gateAND$\label{line:zero:and:start}}{
				\If{$\node.\depth=0\land\parent(\node).agent\neq
						0$}{
					$\node.\agent\gets \parent(\node).\agent$

					$\node.\slot\gets \parent(\node).\slot$

					$\workingSet{}\gets\workingSet{}\setminus\{\node\}$
				}
				\If{$\node.\depth\neq 0$\\
					$\land\forall_{\N\in\children(\node)}
				(\N.\agent\neq 0\lor\N.\depth=0)$}{
					$\varChildNode{}\gets\N\in\children(\node)$ \textbf{s.t.}
					\mbox{\quad}$\forall_{\N'\in \children
						(\node{})} \N.\slot\geq\N'.\slot$

					$\node.\agent\gets \varChildNode.\agent$

					$\node.\slot\gets \varChildNode.\slot$

					$\workingSet{}\gets\workingSet{}\setminus\{\node\}$
				}
					\label{line:zero:and:end}}
		}
\end{algorithm}
Our choice here stems from the \ADT/ gate the node originates from.
We first assign zero-duration nodes to the same agent and time
slot as their parent if the parent is a \gateSEQ node (l.~\ref{line:zero:seq:start}--\ref{line:zero:seq:end}).
\gateNULL, \gateOR and \gateLEAF nodes get the same
assignment as their only child if any, or as their parent if they have no child
(l.~\ref{line:zero:nullorleaf:start}--\ref{line:zero:nullorleaf:end}). The latter case
may happen for \gateNULL{} when handling defences as in \eg{}
Fig.~\ref{fig:pre:nodef:dnok}, and for \gateLEAF{} nodes originally of null duration.
\gateAND nodes are assigned the same agent and time slot as the child that occurs
last (l.~\ref{line:zero:and:start}--\ref{line:zero:and:end}).

Note that in all cases the agents (and time slots) assigned to zero duration nodes are the same as those
of their immediate parents or children. Hence, no further reshuffling is necessary.

%

\begin{proposition}
	\label{prop:algo:zero}
	Adding nodes of zero duration to the assignment in Alg.~\ref{algo:zero} does not affect
	the correctness of the scheduling.
\end{proposition}
\begin{proof}
See~\cite[Proposition~5]{ICECCS2022}.
\end{proof}

\subsection{Complexity and correctness}
\label{sec:algo:proof}


We now consider the algorithm's complexity and prove that it achieves its intended goal.

\begin{proposition}
	\Cref{algo:minsched} is in $\mathcal{O}(kn^2\log n)$, where $k$ is
	the number of input \DAG/s, and $n$ their average number of nodes.
%
		%
	\begin{proof}
	See~\cite[Proposition~6]{ICECCS2022}.
	\end{proof}
\end{proposition}

Thus, while the scheduling algorithm itself is quadratic,
it is executed for $k$ \DAG/ variants, where $k$ is exponential
in the number of $\gateOR$ and defence nodes in the \ADT/.

\begin{proposition}
	The assignments returned by \Cref{algo:minsched} are correct and use the minimal number of agents
	for each variant $\inputDAG \in \DAGset$ to achieve the attack in	minimal time.
	\begin{proof}
	See~\cite[Proposition~7]{ICECCS2022}.
	\end{proof}
\end{proposition}

\subsection{Scheduling for the treasure hunters \ADT/}
\label{sec:algo:ex}

We now apply these algorithms to the treasure hunters example.
\Cref{fig:sched:pruneNlabel} shows the output of the three initial subprocedures.
The depth of nodes assigned by $\textsc{DepthNode}$ is displayed in \textcolor{green!60!black}
{green}.
The branch corresponding to attack \leaf{e} has been pruned 
as per \Cref{sec:pre:or}.
Levels assigned by $\textsc{LevelNode}$ are displayed in \textcolor{blue}{blue}.
Finally, the agents assignment computed by \Cref{algo:minsched} is shown in \Cref{fig:sched:assign}.

\begin{figure}[!!htb]
	\centering
	\scalebox{0.67}{
		\rowcolors{2}{lightgray!30}{white}

		\begin{tabular}{c|l|l}
			\diagbox[]{slot}{agent} & 1                                                & 2                       \\
			\hline
			125                     & \leaf{h_3}, \leaf{GA'}, \leaf{TF'_2}, \leaf{TS'} &                         \\
			124                     & \leaf{h_2}                                       &                         \\
			123                     & \leaf{h_1}, \leaf{h'}                            &                         \\
			122                     & \leaf{ST_2}, \leaf{TF'_1}                        &                         \\
			121                     & \leaf{ST_1}, \leaf{ST'}                          &                         \\
			120                     & \leaf{f_{120}}                                   & \leaf{b_{60}}           \\
			$\cdots$                & $\cdots$                                         & $\cdots$                \\
			61                      & \leaf{f_{61}}                                    & \leaf{b_{1}}, \leaf{b'} \\
			60                      & \leaf{f_{60}}                                    &                         \\
			$\cdots$                & $\cdots$                                         &                         \\
			1                       & \leaf{f_{1}}, \leaf{f'}                          &                         \\
			\hline
		\end{tabular}
	}
		\caption{Treasure hunters: Assignment of \Cref{algo:minsched}}
		\label{fig:sched:assign}
\end{figure}

\section{Experiments}
\label{sec:expe}


The algorithms presented here are implemented in our open source tool \tool~\cite{adt2amas}, written in \texttt{C++17}.
It allows for specifying input \ADT/s either via simple-syntax text files or using
an intuitive GUI,
and handles both their translation to extended AMAS and computation of an optimal schedule with minimal number of agents.
Intermediary steps of the algorithm can be exported as Tikz figures, allowing to easily visualise and understand them.
For more details on the architecture of \tool, we refer the reader to~\cite{ADT2AMASDemoPaper}.
Here, we present its application to the use cases from \cite{ICFEM2020}, plus
examples that feature some specific behaviour. 
All the figures and tables of the examples can be found in the supplementary material of
this paper \url{https://bit.ly/3ONeSzq} and in the extended version of \cite{ICECCS2022}
available at \url{https://arxiv.org/abs/2101.06838}. 

\paragraph*{forestall} This case study models forestalling a software instance. Depending
on the active defences, 4 cases are possible. However, the \DAG/ for no active
defence and the one where the  only active defence is  \leaf{id} (intrusion detection \cite{ICFEM2020}),  are the same. All three remaining
\DAG/s  have an optimal schedule with only 1 agent, in 43 days for the no defence
(or \leaf{id} only) case, 54 if only \leaf{scr} (secure coding rooms) is active, and 55 if both defences
occur. Although only a single agent is needed to achieve the attack in minimal time,
the schedule exhibits which specific attacks must be performed to do so.

\paragraph*{iot-dev} This example models an attack on an IoT device via a network.
There are 4 cases, according to the active
defences, but only the one with no
defence leads to a \DAG/. Indeed, \leaf{tla} (two-level authentication) causes the failure of \gate{GVC}(get valid credentials) which
in turn makes \gate{APN}(access private net) and then \gate{APNS} fail, independent of the defence 
\leaf{inc} (inform of new connections).
Thus the attack necessarily fails. This is also the case if defence \leaf{inc} is
active. The only way for an attack to succeed is that all defences fail, leading to
an optimal schedule in 694 minutes with 2 agents. Hence an attacker will use 2 agents
to perform the fastest attack. On the other hand, the defender knows that a single
one of the two defences is sufficient to block any attack.

\paragraph*{gain-admin} This third case is about an attacker trying to gain
administration privileges on a computer system. There are 16 possible defences
combinations,
which are covered
by only 3 cases: \leaf{scr} (secure coding rooms) is not active; \leaf{scr} is active but not \gate{DTH} (defence against trojans);
both of them are active. In all three cases, the shortest attack requires only a single
agent, and can be scheduled
in 2942, 4320 and 5762 minutes, respectively.


\paragraph*{Exhibiting particular scheduling features} Experiments were conducted
on the example used in~\cite{ICFEM2020} to evaluate the impact of the number of agents on the
attack time, and two small examples designed
to exhibit particular characteristics of the schedule.
Our algorithm confirms an optimal schedule in 5 minutes with 6
agents for the example of~\cite{ICFEM2020}.
Then, \emph{interrupted} (see \Cref{fig:interrupted}) 
shows that the scheduling algorithm can produce an interleaved execution of two
attacks (\leaf{b} and \leaf{e}), assigned to the same agent.
Finally, the \emph{last} example provides a succession of nodes with 0 duration (
\leaf{a'},
\leaf{e'}, \leaf{f'}, \leaf{h'}
and \leaf{i'}), and shows they are handled as expected.

\paragraph*{Scaling example}
In the \emph{scaling} example, the first agent processes the longest path while the
second agent handles all other actions.
It is extended to analyse the scaling capabilities of the scheduling
algorithm. For this purpose, we wrote an automatic generator of \ADT/s.
The parameters of the generated \ADT/s are the \emph{depth},
the \emph{width} corresponding to the number of deepmost leaves,
the number of \emph{children} for each \gate{AND},
and the total number of \emph{nodes}.
All nodes have time 1 except the first leaf that has time $\mathit{width}-1$.
The results show that the number of agents is not proportional to the width of the tree,
and the optimal scheduling varies according to the time of nodes.
{We refer the reader to~\cite{ICECCS2022} for a detailed comparison.}


\section{A general approach with Rewriting Logic}
\label{sec:rewriting}

This section presents an alternative approach for solving the optimal
scheduling problem in ADTrees, which is more general in the sense that it does not build
upon a dedicated algorithm. We start with an appropriate representation for
the ADTree structure (\S \ref{subsec:eq}) and present a rewrite theory giving
meaning to the gates of the tree (\S \ref{subsec:semantics}).
Since the resulting theory is
executable,  we can use the system Maude \cite{DBLP:conf/maude/2007} as a
decision procedure to enumerate all the possible configurations leading to an
attack and find the optimal one (\S \ref{sec:num-agents}). However, without a
suitable strategy, it is not efficient enough for more
complex scenarios. Hence, we refine (\S \ref{subsec:heuristics}) the theory
by adapting some of the ideas and heuristics implemented in the specialised
algorithm proposed in \Cref{sec:algo}. The resulting procedure is easy to prove correct,
and exhibits good performance for all the case studies considered in \Cref{sec:expe}.

In what follows, we explain the main concepts behind Rewriting Logic (RL)
\cite{meseguer-rltcs-1992,meseguer-twenty-2012}, 
while gradually introducing the proposed rewrite theory for ADTrees. We adopt, in
most cases, the notation of Maude \cite{DBLP:conf/maude/2007}, a
high-level language supporting rewriting
logic theories. This allows for producing an executable
specification. For the sake of readability, we omit some details and the complete specification
can be found at the website of our tool \toolM \cite{adt2maude}.

A \emph{rewrite theory} is a tuple $\cR = (\Sigma, E \uplus B, R)$. The static
behaviour (\S \ref{subsec:eq}) of the system  is modelled by the order-sorted
equational theory $(\Sigma, E \uplus B)$ and the dynamic behaviour (\S
\ref{subsec:semantics})  by the set of rewrite rules $R$.

\subsection{Equational theory}\label{subsec:eq} The signature $\Sigma$ defines
a set of typed operators  used to build the terms of the language (\ie{} the
syntax of the modelled system).   $E$ is a set of (conditional) equations over
$T_\Sigma$ (the set of terms built from $\Sigma$) of the form $t = t' ~
  \mathbf{if} \phi$. The equations specify the algebraic identities that terms of
the language must satisfy. For instance, if the operator $|\cdot|$ denotes
the length of a sequence of symbols, then the following equations must
hold:  $|\epsilon| = 0$ and $|ax|  = 1+ |x|$ (where $\epsilon$ is the empty
sequence).

In $(\Sigma, E \uplus B)$, $B$ is a set of structural axioms over $T_\Sigma$ for
which there is a finitary matching algorithm. Such axioms include
associativity, commutativity, and identity, or combinations of them. For
instance,  $\epsilon$ is the identity for concatenation and then, modulo this
axiom, the terms $x\epsilon$ and $x$ are equivalent. The equational theory
associated with $\cR$ thus defines algebraic data types and deterministic and
finite computations as in a functional programming language.

RL allows for defining any syntax for the operators in $\Sigma$,
using \emph{sorts} along with constructors and operators for them.
Here is a simple example defining Peano's natural numbers:

\begin{maude}

fmod NAT is                 --- equational theory
 sort Nat.                  --- sort definition
 op 0 : -> Nat [ctor] .     --- zero
 op s : Nat -> Nat [ctor] . --- successor
 op _+_ : Nat Nat -> Nat .  --- addition
 vars x y : Nat .           --- logical variables
 eq 0 + x = x .             --- equations defining +
 eq s(y) + x = s(y + x) .
endfm
\end{maude}

The attribute \code{[ctor]} in the definition of zero and successor is
optional. It is used to document that these operators are constructors for
terms of sort \code{Nat}. The positions of the
arguments in the (mixfix) operator \code{+} are indicated with underscores
and the equations give meaning to it:  $\forall x:Nat, 0 + x = x$ and $\forall x
  y:Nat, s(y) + x = s(x + y)$. Hence, the term $s(0) + s(s(0))$ reduces to the
normal form $s(s(s(0)))$.

The starting point for our specification is to define an equational theory for
building terms representing ADTrees. In Maude, systems are specified
using a syntax resembling that of object oriented languages.
The needed sorts and operators are defined in the module
\code{CONFIGURATION}, available in  Maude's prelude. The idea is to represent
entities as record-like structures (sort \code{Object}) of the form $\langle O
  : C ~|~ a_1 : v_1,\cdots a_n : v_n \rangle$ where $O$ is an object identifier
(sort \code{Oid}), $C$ is a class identifier (sort \code{Cid}), $a_i$ is an
attribute (sort \code{Attribute}) and  $v_i$ is a term  that represents the
current value of $a_i$. 
We start by defining the class identifiers for each kind of gate:

\begin{maude}

  mod ADTree is          --- Rewrite theory ADTree
  --- Class IDs for Nodes
  op NOT  :    -> Cid .     op AND :      -> Cid .
  op SAND :    -> Cid .     op OR  :      -> Cid .
  op ATK  :    -> Cid .     op DEF :      -> Cid .
\end{maude}

The class \code{NOT} is used to define subtrees that are defences (as in NAND gates);
\code{SAND} stands for sequential \code{AND}; and the last two classes
represent attacks and defences.

The attributes for the gates include the (accumulated) time, cost, and
the number of agents needed to perform the attack:

\begin{maude}

  --- attributes for gates
  op time:_     : Nat      -> Attribute .
  op cost:_     : Nat      -> Attribute .
  op agents:_   : Nat      -> Attribute .
  op acctime:_  : Nat      -> Attribute .
  op acccost:_  : Nat      -> Attribute .
\end{maude}

The equational theory is ordered-sorted, \ie
there is a partial order on sorts defining a
sub-typing relation: \code{subsort Qid < Oid .}
The sort \code{Qid} is part of Maude's standard library and represents
quoted identifiers, \eg{} \code{'TS} (a sequence of characters preceded by an
apostrophe).  Hence, \code{'TS} is both a quoted identifier and an object identifier.

An interesting RL feature is the definition of
axioms for the operators ($B$ above), \eg it is straightforward to define
a list as a non-commutative monoid and a set as an abelian
monoid:
\begin{maude}

subsort Oid < List .        --- singleton list
subsort Oid < Set .         --- singleton set
op nil : -> List [ctor] .   --- empty list
op empty : -> Set [ctor].   --- empty set
--- building lists and sets
op __ : List List -> List [ctor assoc id: nil] .
op _,_ : Set Set -> Set [ctor assoc comm id: empty] .
\end{maude}

In this specification, the term  ``\code{'A 'B 'C}"  (resp. ``\code{'A, 'B,
  'C}") represents a list (resp. a set) of three object identifiers. The
concatenation operator \code{__} is called \emph{empty syntax}, since a white
space is used to concatenate elements. Note that being associative, the lists
``\code{'A ('B 'C)}" and ``\code{('A 'B) 'C}" are equivalent (modulo
\code{assoc}), as are the terms ``\code{'A, 'B, 'C}"
and ``\code{'C, 'B, 'A}" due to commutativity.

The sorts and operators needed to specify lists and sets are already available
in Maude. The sorts for these data structures are renamed here, respectively,  as
\code{NodeList} and \code{NodeSet} and used below to define two new attributes
for gates:

\begin{maude}

  --- ordered and unordered children
  op lchd:_ : NodeList -> Attribute .
  op schd:_ : NodeSet  -> Attribute .

\end{maude}

The first one is used for sequential gates \code{SAND},
and the second one for all others.
Each node is associated with a state:

\begin{maude}

  --- states for nodes in the tree
  sort Status .
  ops Fail Succeed Unknown  :        ->    Status .
  op stat:_                 : Status -> Attribute .
\end{maude}

Initially, all the nodes are in state \code{Unknown}, which may change to
\code{Succeed} or \code{Fail},
according to the rules described in the next section.

Suitable operators for building the different gates in an ADTree are
introduced. For instance:



\begin{maude}

--- building an attack: ID, time and cost
op makeAtk : Qid Nat Nat -> Object .
eq makeAtk(Q, t, c) =
   < Q : ATK | time: t, cost: c, agents: 1 ,
               acctime: 0, acccost: 0,
               stat: Unknown > .
--- build. an OR gate: ID, children, time and cost
op makeOr : Qid NodeSet Nat Nat -> Object .
eq makeOr(Q, S, t, c) =
    < Q : OR | time: t , cost: c, agents: 0 ,
               acctime: 0, acccost: 0 ,schd: S,
               stat: Unknown > .
\end{maude}

Note that a leaf attack requires one agent and the number of agents for
the \code{OR} gate is initially zero. That value will be updated as explained
below.

An equational theory is executable only if it is
terminating, confluent and sort-decreasing \cite{DBLP:conf/maude/2007}. Under
these conditions, the mathematical meaning of the equality $t \equiv t'$
coincides with the following strategy: reduce $t$ and $t'$ to their unique (due
to termination and confluence) normal forms $t_c$ and $t'_c$ using the
equations in the theory as \emph{simplification rules} from left to right.
Then, $t \equiv t'$ iff $t_c =_B t_c'$ (note that $=_B$, equality modulo $B$,
is decidable since a finitary matching algorithm for $B$ is assumed). For
instance, the term \code{makeAtk('A, 3 ,2)} can be reduced to the normal form
\lstinline[mathescape]!<'A : ATK | time: 3, agents: 1, stat: Unknown, ...> !
using the equations above.

The Maude's theory \code{CONFIGURATION} defines the sort
\code{Configuration} as a set of objects concatenated
with the empty syntax (an associative and commutative operator with
\code{none} as identity). Hence, the term $t_{GA}$ below,
with sort \code{Configuration}, encodes the subtree \code{GA} in
\Cref{fig:adt:treasure}.

\begin{maude}

  --- t_GA (subtree GA)
  MakeAtk('h,3,500) MakeAtk('e, 10,0)
  MakeOr('GA, ('h, 'e), 0, 0)

\end{maude}

Finally, two additional constructors for the sort \code{Configuration} are defined
in the theory \code{ADTree}:

\begin{maude}

  op {_;_} : Oid Configuration  -> Configuration .
  op {_}   :     Configuration  -> Configuration .

\end{maude}

Given an ADTree $T$, we shall use $\os T \cs$ to denote the corresponding  term
of the form  \code{\{Q,Cnf\}} where  \code{Q} is the root of $T$ and
\code{Cnf} is the set of objects encoding the gates in $T$. As shown in the next
section, the second operator \code{op \{_\}} will be useful to simplify the final configuration
and summarise the results of the analysis.

\subsection{Rewriting semantics for gates}\label{subsec:semantics} Now we focus
on the last component $R$ in the rewrite theory $\cR = (\Sigma, E \uplus B,
  R)$. This is a finite set of of conditional rewriting rules of the form
\(\crl{l(\olist{x})}{r(\olist{x})}{\phi(\olist{x})}\),  specifying a pattern
$l(\olist{x})$ that can match some fragment of the system's state $t$ if there
is a substitution $\theta$ for the variables $\olist{x}$ that makes
$\theta(l(\olist{x}))$ equal (modulo axioms) to that state fragment. If the
condition $\theta(\phi(\olist{x}))$ is true, the new state fragment is
$\theta(r(\olist{x}))$, leading  to a local transition. Hence, rules  define
state transformations modelling the dynamic behaviour of the system (which is not
necessarily deterministic, nor terminating).

Conditions and patterns in rules may considerably affect the performance of a
rewrite theory when it is used to explore all the possible reachable states
from a given term. In this section, we propose rules that are self-explanatory
but that may exhibit unnecessary non-determinism during the search procedure.
Later, we add extra conditions to reduce the search space and improve
efficiency.

\noindent\textbf{Leaves.} Let us start defining the behaviour for the gates
representing leaves of an ADTree, i.e., attacks and defences:

\begin{maude}

--- semantics for attacks
rl [ATKOK] :  < Q : ATK | stat: Unknown, ats > =>
              < Q : ATK | stat: Succeed, ats >  .
rl [ATKNOK]:  < Q : ATK | stat: Unknown, ats > =>
              < Q : ATK | stat: Fail,    ats >  .
\end{maude}

These are unconditional ($\phi = true$) rules and then, $\phi$ is omitted.
\code{Q} (resp. \code{ats}) is a logical variable of sort \code{Qid} (resp.
\code{AttributeSet}, a set of attributes). These rules change the state of an
attack currently in state \code{Unknown} to either \code{Succeed} or
\code{Fail}. For instance, consider the term $t_{GA}$ (of sort
\code{Configuration}) above. Due to the structural axioms governing the
juxtaposition operator (\code{[assoc comm id: none]}), these two rules can be applied in
two different positions (local fragments) of the system represented by
$t_{GA}$. More precisely, the rules \code{[ATKOK]} and \code{[ATKNOK]} can be
applied by either substituting the variable \code{Q} with the term \code{'h}
(and \code{ats} with \code{time: 3, cost: 500,...}) or substituting
\code{Q}
with \code{'e}. Hence, the term $t_{GA}$ can be rewritten in two steps into
four possible configurations where: both attacks fail, one of the attacks
succeeds and the other fails, or both attacks succeed. That is, all the
possible outcomes for the attacks are covered.

The rules for defences are defined similarly:
\begin{maude}

--- semantics for defences
rl [DEFOK] :  < Q : DEF | stat: Unknown, ats > =>
              < Q : DEF | stat: Succeed, ats >  .
rl [DEFNOK] : < Q : DEF | stat: Unknown, ats > =>
              < Q : DEF | stat: Fail,    ats >  .
\end{maude}

\noindent\textbf{Gates.} Let us start with the rules for the \texttt{OR} gate:

\begin{maude}

rl [OR] :
< Q : OR | schd: (o, S), stat: Unk., used: U, ats>
< o : C  |               stat: Succ., ats' >  =>
< Q : OR | schd: empty , stat: Succ., used: (U,o),
                         accumulate(ats,ats') >
< o : C  |               stat: Succ., ats' >
\end{maude}



The left-hand side (LHS) of the rule matches a fragment of the global system containing
two objects: an \code{OR} gate and an object \code{o} of any class (\code{o}
and \code{C} are variables of sort \code{Oid} and \code{Cid}
respectively). The term \code{(o, S)}, where \code{S} has sort
\code{NodeSet}, is a set. Hence, this rule applies to any of the children (in
state \code{Succeed}) of the gate. The right-hand side (RHS) dictates the
new state: the \code{OR} gate moves to the state \code{Succeed}; the node
\code{o} is added to the attribute \code{used}, witnessing that  \code{o} is
required to perform the attack \code{Q}; and the attributes for time, cost and
the number of agents in \code{o} are accumulated in \code{Q}. This is the
purpose of the function \code{accumulate} that computes the new values
from  the attributes  of \code{Q} (\code{ats}) and those of \code{o} (\code{ats'}).
The new values for time and cost result from adding the time and cost
accumulated in the children \code{o} with the time and cost of the gate
\code{Q}. Moreover, the number of agents needed to
perform \code{Q} is set to the number of agents needed to perform \code{o}.
This is an upper bound for the number of agents needed, where one of the agents
working on the subtree \code{o} can complete \code{Q}.

Now we consider two rules for handling the cases when one of the children
of the \code{OR} gate fails and where there are no more children to be
considered:

\begin{maude}

  rl [OR] :  --- failing child
  < Q : OR | schd: (o, S), stat: Unknown,  ats >
  < o : C  |               stat: Fail,     ats' > =>
  < Q : OR | schd: S,      stat: Unknown,  ats >
  < o : C  |               stat: Fail,     ats' > .
  rl [OR] :  --- no more children
  < Q : OR | schd: empty, stat: Unknown, ats > =>
  < Q : OR | schd: empty, stat: Fail,    ats > .
\end{maude}

The first rule discards a failing child of the \code{OR} gate. The second rule
changes the state of the gate to \code{Fail} when there are no remaining
children. With these rules, the term $t_{GA}$ can be rewritten into three
possible configurations where the gate \code{GA}: fails (when both \code{h} and
\code{e} fail); succeeds with total time $3$ (when \code{h} succeeds,
regardless the state of \code{e}); and succeeds with total time $10$ (when
\code{e} succeeds).

The rules for the (parallel) \code{AND} gate are defined as follows:

\begin{maude}

rl [AND] : --- succeeded child
 < Q : AND|schd: (o, S), stat: Unk.,  used: U, ats>
 < o : C  |              stat: Succeed,  ats'  > =>
 < Q : AND|schd: S,      stat: Unknown,  used:(U,o),
                         acc-max(ats,ats') >
 < o : C  |              stat: Succeed,  ats'  > .
rl [AND] :  --- failing child
 < Q : AND | schd: (o, S), stat: Unknown,  ats >
 < o : C   |               stat: Fail,     ats' > =>
 < Q : AND | schd: empty , stat: Fail,     ats > .
rl [AND] :  --- no more children
 < Q : AND | schd: empty, stat: Unknown, ats > =>
 < Q : AND | schd: empty, stat: Succeed, ats > .
\end{maude}

In the first rule, the operator \code{acc-max} accumulates the \code{time} attribute  by using the
function \code{max}. That is, the \code{AND} gate computes the maximal value
among the time needed to perform the attacks in each of the children of
\code{Q}. On the contrary, the number of agents is accumulated by adding the
value of the attribute \code{agents} of \code{o} and \code{Q}. Intuitively,
since the children of \code{Q} can be executed in parallel (and in any order),
an upper bound for the number of agents needed in \code{Q} is the sum of the
agents needed for each of \code{Q}'s children.
In the second rule, as expected, a failure of one of the children implies the
failure of the gate.  In the third rule, when all the children
succeed (and \code{schd} is empty) the gate succeeds.

The behaviour of the sequential gate is specified as follows:

\begin{maude}

rl [SAND] :
 < Q : SAND|lchd: (o L), stat: Unk., used: U, ats >
 < o : C   |             stat: Suc.,       ats'> =>
 < Q : SAND|lchd: L    , stat: Unk., used: (U,o)
                         accumulate(ats, ats') >
 < o    : C|             stat: Suc.,  ats' > .
\end{maude}

The term \code{(o L)} is a list and this rule only
matches a state where the first child of the gate is in state \code{Succeed}.
Similar rules to those presented for the \code{AND} gate handling the cases for
a failing child and an empty list of children are also part of the
specification and omitted. The attribute time is accumulated in this case
by adding the values in $o$ and $Q$. For the number of agents, the value is
accumulated using the function \code{max}: the attack is sequential and the
number of agents needed in $Q$ is bound by the child that requires more agents.

The next rules give meaning to the \code{NOT} gate, used to
model the gates \code{CAND}, \code{NODEF} and \code{SCAND}
in Figure \ref{fig:adt:constructs}:

\begin{maude}

rl [NOT] :
 < Q : NOT | lchd: o, stat: Unknown, ats  >
 < o : C   |          stat: Succeed, ats' > =>
 < Q : NOT | stat: Fail,    acc-def(ats) >
 < o : C   | stat: Succeed, ats' > .
rl [NOT] :
 < Q : NOT | lchd: o, stat: Unknown, ats  >
 < o : C   |          stat: Fail,    ats' > =>
 < Q : NOT | stat: Succeed, acc-def(ats) >
 < o : C   | stat: Fail,    ats' > .
\end{maude}

As expected, if the (unique) child of a \code{NOT} gate succeeds, the gate
fails and vice-versa. The time, cost and number of agents are  accumulated in a
different attribute (\code{acc-def}) since those correspond to the resources
for a defence (and not for an attack).

We add an extra rule whose unique purpose is to summarise the results of the
analysis:

\begin{maude}

rl [END] :
 {Q ;< Q : C | stat: Succeed,
               agents: a, acctime: t, ats > Cnf } =>
{    < Q : C | agents: a, acctime: t >
     < gates: attacks((< Q : C | ats > Cnf )) >
     < defences: act-defences(Cnf) > } .
\end{maude}

This rule is enabled only when the root of the tree \code{Q} is in state
\code{Succeed}. All the attributes but the accumulated time and the number of
agents are discarded. The nodes of the tree (\code{Cnf}) but the root are also
discarded. Two new objects are created, namely  \code{gates} and
\code{defences}, that store the set of attacks and defences
enabled in the final configuration. Such sets are computed with the aid of the
operators  \code{attacks} (that uses the attribute \code{used} in the gates)
and \code{act-defences}. Note that the shape of the configuration has changed,
from \code{\{Q;Cnf\}} to \code{\{Cnf\}} (see the operators defined in the end of \Cref{subsec:eq}).\\

\noindent\textbf{Exploring the search space.} A rewrite theory  $\mathcal{R}$
proves sequents of the form $\mathcal{R} \vdash t \redi^{*} t'$ meaning that
the term $t$ rewrites in zero or more steps into $t'$. Here, we
are interested in proving sequents of the form $\mathcal{R} \vdash t \redi!~
  t'$ meaning that $t\redi^{*}t'$ and $t'$ cannot be further rewritten.

Let us call $\RADT$ the rewrite theory defined above that represents the
state of an ADTree and its execution. For an ADTree $T$, if
$\RADT \vdash \os T\cs \redi!~ t'$ then
$t'$ can be either a configuration
where the root node $Q$ fails (and the other gates are in a state different
from \code{Unknown}) or a term of the form
\[\small
  \begin{array}{ll}
    \{\conf{Q : C \mid agents: a, acctime: t}
    \conf{\textit{gates}: SA} \conf{\textit{defences}: SD} \}
  \end{array}
\]

where $a$ and $t$  are, respectively,  the upper bound for the number of agents
and the time needed to perform the root attack $Q$. Moreover, $SA$ and $SD$
are, respectively, the set of enabled attacks and defences in the final
configuration. For now on, the term above will be written as  $[a,t,SA, SD]$.

\begin{example}\label{ex:rew}

  Let $T$ be the ADTree in Figure \ref{fig:adt:treasure} and $t_{TS} = \os
    T\cs$.   Using the above defined rewrite theory, the Maude's command
  \code{search t-TS =>! Cnf:Configuration} finds four (distinct) final
  configurations corresponding to the two possible outcomes of the defence $p$
  and the choice of the attack used in the gate $GA$. In the two non-failing
  configurations, $p$ is not enabled. In one of them, $h$ is chosen and the
  total time for the attack is $125$. In the other, $e$ is executed with total
  time $132$.

\end{example}

\begin{theorem}[Correctness]\label{th:correct}

  Let $T$ be an ADTree. Then, $\RADT \vdash \os T \cs \redi! ~ [a,t,SA, SD]$ iff
  there is an attack in $T$ of time $t$ where the attacks (resp. defences) in $SA$
  (resp. $SD$) are enabled.

\end{theorem}
\begin{proof}

  ($\Rightarrow$) We must have  $\os T \cs \redi^{*} t' \redi [a,t,SA,SD]$
  where the last rule applied is necessarily  \code{[END]}. Consider the
  derivation $\os T \cs \redi^{*} t'$ where the rules for the different gates
  are applied. An invariant in each step of such a derivation is that when the
  accumulated time attribute is modified in a gate, it is computed correctly. For
  instance, when the rule \code{[SAND]} is applied, the accumulated time is the
  sum of the time needed to perform each of the children of the gate.
  Following the rules applied in the derivation, we can reconstruct the attack
  in $T$.

  \noindent ($\Leftarrow$) Consider the particular sequence of rewriting where the
  rule \code{[ATKOK]} is applied in all the attack leaves in $SA$ and \code{[ATKNOK]}
  in the others. Similarly for the defences in $SD$. This completely
  determines the way the rules for the gates need to be applied, thus
  reproducing the same attack.
\end{proof}

As illustrated in Example \ref{ex:rew}, we can use the search facilities
in Maude to list all the final (successful) configurations to perform an
attack and find the minimal time.

\begin{theorem}[Optimal time]
  Let $T$ be an ADTree. If the minimal time to perform the main attack in $T$
  is $t$, then there exists $a$, $SD$ and $ST$ s.t.
  $\mathcal{R} \vdash \os T\cs \redi{!} ~[a,t,SA,SD]$.
\end{theorem}
\begin{proof}
  Immediate from \Cref{th:correct}.
\end{proof}

Unfortunately, this procedure does not allow for finding the minimal
number of agents but only an upper bound for it. The reason is that the
operator \code{acc-max} (see rule \code{[AND]}) sums the
number of agents needed for each child of the gate.
Hence, for instance, this procedure determines that the number
of agents to perform the attack in  \Cref{fig:interrupted}
is $3$ (two agents to perform concurrently $d$ and $e$ and an extra one to perform $b$).
However, there is an attack using only 2 agents (\Cref{ex:search2}). The key point
is that the semantics does not handle the case where an
agent can be shared between different branches of the tree.

\begin{theorem}[Upper bound for the number of agents]
  Let $T$ be an ADTree. If $n$ agents can perform an attack on $T$ with time $t$,
  then there exists $a$, $SA$ and $SD$ s.t.
  $\mathcal{R} \vdash \os T\cs \redi^{*} [a,t,SA,SD]$
  and $n \leq a$.
\end{theorem}

\begin{proof}

  Similar to the proof of \Cref{th:correct}.
\end{proof}

\subsection{Minimal set of agents}\label{sec:num-agents}

This section proposes a second rewrite theory $\RADTA$ useful for finding the
minimal set of agents to perform an attack. The starting point is a new
constructor for the sort \code{Configuration} with the following attributes:

\begin{maude}

{agents:_      --- schedule
 global-time:_ --- elapsed time
 max-time:_    --- max time for the attack
 enabled:_     --- set of enabled attacks
 disabled:_    --- atks. that cannot be performed now
 system:_      --- representation of the system/gates
}
\end{maude}

The first attribute is a list of terms of the form \code{[L] :: N} where
\code{L} is a list of node identifiers and \code{N} a natural number.
The term \code{([ a b ] :: 3) ( [ c ] :: 0 )} represents a
scenario with two agents: the first one has already performed the attack $a$
and she is currently working on $b$ with remaining duration $3$;
and the
second agent has already performed $c$ and she is currently free ($N=0$).
The attribute \code{global-time} is a global clock indicating the current
time-unit. \code{max-time} is the maximal time the agents have to perform the
attack, and its value will be initialised with the time computed with the theory
$\RADT$. The set $SA$, computed by $\RADT$, is partitioned into two sets,
namely, \code{enabled}  and \code{disabled}. All the non-leaf gates are in the
second set as well as the leaves which belong to a subtree that is not the
first child of a sequential gate. The other (leaf) attacks are in the set
\code{enabled}. The last attribute stores the representation of the ADTree
($\os T\cs$).

The following operator will be useful to build the initial configuration:

\begin{maude}

op make-schedule : Nat Nat NodeSet Conf -> Conf .
ceq make-schedule(n, t , S , Sys) =
  { agents: make-agents(n) --- build. the list of ag.
    global-time: 0
    max-time: t
    enabled: intersection(S', S)
    disabled: S \ S'       --- set difference
    system: Sys
} if S' := all-attacks(Sys) .
\end{maude}
where the first two parameters are, respectively, the number of agents
and the total time for the attack. The third parameter is the set of
enabled attacks and the last parameter the representation of the
ADTree.

In what follows, we define rules to non-deterministically assign
attacks to agents, move attacks from the set \code{disabled} to the set
\code{enabled} and make the global time advance.
Let us start with the rule assigning an attack to an agent. For the
sake of readability, the parts of the configuration not modified by the rule
are omitted:

\begin{maude}

rl [pick] :
{ agents: SL ([L] :: 0) SL'
  enabled: (o, S)
  system: { Q ;  Cnf < o : C | ats , time: t >} } =>
{ agents: SL ([L o] :: t) SL'
  enabled: (S)
  system: { Q ;  Cnf < o : C | ats , time: t >} }  .
\end{maude}

One of the enabled attacks \code{o} is
assigned to a free agent (\code{SL} and
\code{SL'} are lists of terms of the form \code{[L]::N}). After the
transition, the chosen agent is working on \code{o} with duration \code{t}.

It is also possible for an agent to interrupt the current attack she is working
on and pick another (enabled) attack. This is the purpose of the following
rule:

\begin{maude}

rl [inter] :
{ agents: SL ([L o ] :: nt) SL'
  enabled: (o', S)
  system: { Q ;  Cnf < o : C   | ats,  time: t >
                     < o' : C' | ats', time: t' > }
} =>
{ agents: SL ([L o' ] :: t') SL'
  enabled: (o, S)
  system: { Q ;  Cnf < o : C   | ats,  time: nt >
                     < o' : C' | ats', time: t' > }}
\end{maude}

After the transition: the attack \code{o} is back to the set \code{enabled};
the remaining time for \code{o} is updated to \code{nt} in the attribute
\code{system}; and the attack \code{o'} with duration \code{t'} is scheduled.

The next rule models the fact that the time advances for all the (busy) agents:

\begin{maude}

  rl [time] : {agents: SL   global-time: n  }
  => {agents: minus(SL, 1)  global-time: n+1} .
\end{maude}

The function \code{minus} simply decrements by $1$ the time needed to finish
the current task for each busy agent. Since time advances by one unit
and agents are free to interrupt their current task, these rules effectively
model the preprocessing proposed in \Cref{sec:preprocess}.
Now, consider the two rules below:

\begin{maude}

rl [END]: {agents: SL      global-time: n
           enabled: empty  disabled: empty 
           system: { Q ; Cnf <Q:C|stat: Suc., ats >}}
       => { agents: SL } .
crl [FAIL]: { global-time: n max-time: n' } => fail
            if n > n' .

\end{maude}

The rule \code{[END]} finishes the computation when the root of the ADTree is
in state \code{Succeed} and  there are no more pending
attacks to be executed.
The second rule is conditional: if the global time $n$ is greater
than the maximal time $n'$, then the configuration reduces to \code{fail}.
That is, the agents could not meet the deadline for the attack.

To conclude, we introduce rules governing the movement
of attacks between the sets \code{enabled} and \code{disabled}:

\begin{maude}

rl [done] :
{ agents: SL ([L o] :: 0) SL'
  system: { Q ; Cnf < o : C | ats , stat: Unknown > }
} =>
{ agents: SL ([L o] :: 0) SL'
  system: { Q ; Cnf < o : C | ats , stat: Succeed>}} .
rl [active] :
{ enabled: S         disabled: (o, S')
  system: {Q ; Cnf 
           < o:SAND | ats , stat: Unk., lchd: nil > }
} =>
{ enabled: (o, S)    disabled: S'
  system: {Q ; Cnf 
           < o:SAND | ats , stat: Unk., lchd: nil > }
} .

\end{maude}

If an agent has already finished the attack \code{o}, the rule \code{[done]}
updates the state of \code{o} from \code{Unknown} to \code{Succeed}. The second
rule enables the attack \code{o} when it is a sequential gate whose children
have all already been performed (\code{lsch=nil}). Similar
rules are introduced for the other gates.

\begin{figure}[!!htb]
\vspace{-0.45cm}
\centering
\begingroup
\scalebox{0.9}{
  \begin{tikzpicture}
    \node at (0,3)   {\scalebox{.5}{

\begin{tikzpicture}
	[every node/.style={ultra thick,draw=red,minimum size=6mm},
	node distance=1.8cm]

	\node[and gate US,point up,logic gate inputs=nn] (a)
		{\rotatebox{-90}{\leaf{a}}};

	\node[and gate US,point up,logic gate inputs=nn,
			below left of = a] (c)
		{\rotatebox{-90}{\leaf{c}}};
	\draw (a.input 2) -- ([yshift=-0.15cm]a.input 2) -| (c.east);

	\node[state, below left of = a] (b) {\leaf{b}};
	\draw (a.input 1) -- ([yshift=-0.15cm]a.input 1) -| (b);

	\node[state, below left of = c] (d) {\leaf{d}};
	\draw (c.input 1) -- ([yshift=-0.15cm]c.input 1) -| (d);

	\node[state, below right of = c] (e) {\leaf{e}};
	\draw (c.input 2) -- ([yshift=-0.15cm]c.input 2) -| (e);
\end{tikzpicture}}};
    \node at (0,0) [scale=2]{\scalebox{.555}{

\begin{tabular}{l@{~}r}
	\textbf{Name} & \textbf{Time}\\
	\hline
	\leaf{a} &  0~m\\
	\leaf{b} &  2~m\\
	\leaf{c} &  1~m\\
	\leaf{d} &  4~m\\
	\leaf{e} &  3~m\\
\end{tabular}
}};
		\node at (5,1) [text width=5.7cm,align=justify,scale=1/0.9,draw=none]
		{
			\begin{example}\label{ex:search2}

				Consider the ADTree in \Cref{fig:interrupted}. The $\RADT$ theory
				determines that the attack can be performed in $5$ time-units with at most
				$3$ agents. Starting from a configuration where the attribute \code{agents}
				is set to \code{ ([nil]::0) ([nil]::0)([nil]::0)} and  \code{max-time} to
				$5$, we can enumerate all the possible schedules leading to the attack. One
				of these includes the configuration
				\code{(['d 'c 'a]:: 0)} \code{(['e 'b]:: 0)} \code{([ ] :: 0)},
				where the third agent was not assigned any attack.

			\end{example}
		};
	\end{tikzpicture}
}
\endgroup
  \caption{Interrupted schedule example}
  \label{fig:interrupted}
\end{figure}

%
%

In what follows, we use $[n,t,S,T]$ to denote the term
\code{make-schedule(n,t,S,}$\os T \cs$ \code{)}
and $[SL]$ to denote the term \code{\{agents: SL\}}
(see the RHS in rule \code{[END]}).

\begin{theorem}[Correctness]
  Let $T$ be an ADTree.
  $\RADTA \vdash [n,t,S,T] \redi! ~[SL]$ iff
  there is an attack in $T$ with $n$ agents and time $t$
  where all the attacks in $S$ are performed.
\end{theorem}

\begin{proof}

  As in  \Cref{th:correct}, the close correspondence of steps in the
  attack and rules in $\RADTA$ allows us to rebuild the attack in
  $T$ from the derivation in $\RADTA$ ($\Rightarrow$) and
  vice-versa ($\Leftarrow$).
\end{proof}

\subsection{Heuristics and strategies}\label{subsec:heuristics}

As illustrated  in \Cref{ex:rew,ex:search2}, it is
possible to explore the reachable state space generated from a given term.
The \code{search} command uses a breadth-first strategy: for each node of
the search tree, all the rules, with all possible
matchings, are applied to produce the next level in the search tree.
This guarantees completeness: if $\mathcal{R} \vdash t \redi^* {t'}$ then the
search command will eventually find~$t'$.

The search space generated by terms in the theories $\RADT$ and $\RADTA$ is
certainly finite but it can grow very fast, especially in $\RADTA$. Hence, for
more complex ADTrees, the search procedure will not terminate in a reasonable
time. In this section we show how to control the non-determinism in the
proposed theories. The result is a decision procedure that can be effectively
used in the case studies presented in \Cref{sec:expe}.

\noindent\textbf{Strategy for $\RADT$}. By inspecting the rules in the theory $\RADT$, we can observe that there are
different sources of non-determinism that can be controlled (without losing
solutions). For instance, the last two rules  for the \code{OR} gate (failing
child and no more children) can be eagerly applied: any interleaving with those
rules will produce the same effect. Note that this is not the case for the
first \code{OR} rule: different choices for matching the pattern \code{(o, S)}
produce different results and all the possibilities need to be explored.
Now consider the rules for the (parallel) \code{AND} gate. A failing child
implies the failing of the gate, regardless of the state of the other children.
Moreover, given two children in state \code{Succeed}, it is irrelevant
which one is considered first in an application of the first rule (pattern
\code{(o,S )}). This is the case since function \code{act-max} accumulates values using
\code{+} and \code{max}, both commutative operations.

Now let us explore the rules for the nodes in the leaves of the ADTree.
Consider \code{[ATKOK]} and \code{[ATKNOK]} and the gate \code{GA} in
\Cref{fig:adt:treasure}. This attack succeeds only if either $h$ or $e$
succeeds. If both succeed, the \code{[OR]} rule discards one of them. In other
words, when the rule \code{[OR]} is applied, the status of the discarded
children \code{S} in the pattern \code{(o, S)} is irrelevant and we can safely
assume that the attacks in the subtree \code{S} were not performed. This means
that we can dispense with the application of \code{[ATKNOK]} and rely on the
rule \code{OR} to explore all the possible configurations. Also,
the rules for defences are both needed: the activation or not of a defence
limits the attacks that can be accomplished.

Strategies \cite{DBLP:journals/entcs/Marti-OlietMV09} provide a mechanism for
controlling the way rules are applied in a given theory. In Maude, this is
implemented with the help of a strategy language that tells the rewriting
engine how to explore the state space. The command \code{srew T using STR}
rewrites the term \code{T} according to the strategy expression \code{STR} and
returns all its possible results. 

The basic building block in the strategy language is the application of a
rule. For instance, the command \code{srew T using OR} will apply the rule
\code{[OR]} in  all possible ways on term \code{T}. As discussed above, if
there are different matchings for the application of \code{[AND]}, all of them
lead to the same result. The strategy \code{one(AND)} applied to a term
\code{T} succeeds if \code{[AND]} matches, possibly in different ways, but
only one matching is considered and the others discarded.

Strategies can be defined
by using constructors similar to regular expressions (see the
complete list in \cite[Section 4]{DBLP:journals/jlap/DuranEEMMRT20}):
\code{idle} (identity); empty set / no solution (\code{fail});
concatenation ($\alpha;\beta$); disjunction ($\alpha \mid \beta$);
iteration ($\alpha^*$); conditional application, $\alpha~?~\beta~:~\gamma$,
where $\beta$ is applied on the resulting terms after the application of
$\alpha$, or $\gamma$ is applied if $\alpha$ does not produce any result.
From these, it is possible to define: $\alpha$~\code{or-else}~$\beta$
that executes $\beta$ if $\alpha$ fails; and the normalisation operator
$\alpha!$ that applies $\alpha$ until it cannot be further applied.

Consider the following strategy:
\begin{maude}

  deter := ( one(ATKOK) or-else one(NOT) or-else
             one(ORD) or-else one(SAND) or-else
             one(PAND) ) ! .
  solve := ( ChoiceOK | ChoiceNOK ) ! ;
\end{maude}


where \code{ORD} refers to the second and third rules for the \code{OR} gate.
The strategy \code{deter} applies the confluent rules until a
fixed point is reached (!).
The strategy \code{solve} first explores all the configurations for the defences (active or
inactive). Then, the confluent rules are eagerly applied.
Next, if the \code{[END]} rule can be applied,
the computation finishes:  the  rules for gates do not apply on the resulting term
on the RHS of \code{[END]}, and a further application of \code{deter} necessarily fails. If this is not the
case, the \code{[OR]} rule is tried. If there are no more \code{OR} gates
in the configuration, the strategy fails.
Otherwise, the \code{OR} rule is applied (considering all possible
matchings) and the confluent rules are used again.

Recall that final/irreducible configurations can be either  $\{C\}$ (RHS in
\code{[END]}) or $\{Q;C\}$ where
the gate $Q$ is in state \code{Fail} and all the other rules are in a state
different from \code{Unknown}.

\begin{theorem}[Completeness]
  $\RADT \vdash \{Q ; C\}  \redi!~ \{C'\} $ iff the
  configuration $\{C'\}$ is reachable from $\{Q ; C\}$ following the
  strategy \code{solve}.
\end{theorem}
\begin{proof}
  As explained above, the final outcome of the attack depends on the
  defences and the choices in \code{OR} gates (if any). Consider the rules applied
  in the derivation $\{Q ; C\}  \redi!~ \{C'\}$ (where the last one is
  necessarily \code{[end]}). The activation or not of a defence does not
  depend on any other action (leaf nodes). Hence, we can permute the
  application of those rules to be done at the  beginning of the derivation.
  Due to the commutativity of the operations for accumulating
  values (\code{max} and \code{+}), we can also rearrange the application of
  the rules (except \code{[OR]}) following \code{deter}. 
  Note that the rule \code{[ATKNOK]} may appear in the derivation. However,
  this is only possible in the scope of a subtree  discarded in an \code{OR} gate.
	Hence, we still have a valid derivation without using that rule.
\end{proof}

\noindent\textbf{Non-determinism in  $\RADTA$.} Now let us consider the theory
$\RADTA$ that exhibits many sources of non-determinism. The \code{[pick]} rule
can select any enabled element \code{o} and schedule it for any free (\code{[L]
  :: 0}) agent. Since in the current model agents have the same abilities  to
perform any of the attacks, we may impose an additional restriction in this
rule: all the agents in the list \code{SL} must be working (remaining time different from zero).
Hence, \code{[pick]} will schedule \code{o} to the first free
agent in the list, thus eliminating some (unnecessary) choices.

The rules \code{[pick]} and \code{[time]} can be interleaved in many ways.
One might be tempted to restrict the application of \code{[time]} to
configurations where either there is no enabled activity or where all the
agents are busy. Let us call this strategy PBT (\code{[pick]} before
\code{[time]}). Since $\RADT$ computes an upper bound for the number of agents,
the strategy PBT cannot be used to compute the minimal set of agents: it will
enforce the use of all of them.

An approach to circumvent the problem above is the following. Assume that for
a given ADTree, $\RADT$ finds an attack with a number of agents $n$.  Then,
execute $\RADTA$ with the strategy PBT with a configuration of $i$ agents,
iterating $i$ from $1$ to $n$. The first value for $ i \in 1..n$ that succeeds
will correspond to the optimal number of agents.
The easiest way to enforce PBT is by adding an extra condition to \code{[time]}:
\begin{maude}

crl [time] :
{ agents: SL global-time: n enabled: S } =>
{ agents: minus(SL, 1)      global-time: n + 1
  enabled: S }
if all-busy(SL) or (some-not-busy(SL) and S == empty)
\end{maude}

Hence, time advances only if all the agents are currently working or
the set of enabled attacks is empty. 

There is one extra source of non-determinism that we can control. The
\code{[pick]} rule, in its current form,  can choose any of the enabled
activities. How can we guide such a choice? The answer is in the algorithm in
\Cref{sec:algo}: choose by levels and prioritising the activities with higher
depth. Based on the level and depth, we can define the strict lexicographical
total order $(l,d,id) \prec (l',d',id')$ iff $l<l'$ (first nodes  with higher
levels); or $l=l'$ and $d<d'$ (priority to higher depth); or $l=l'$, $d=d'$ and $id
  <
  id'$ (needed to break ties on activities with the same level and depth). Hence,
the rule \code{[pick]} becomes: 

\begin{maude}

crl [pick] :
{ agents: SL ([L] :: 0) SL'
  enabled: (o, S)
  system: { Q ;  Cnf < o : C | ats , time: t >}
} =>
{ agents: SL ([L o] :: t) SL'
  enabled: (S)
  system: { Q ;  Cnf < o : C | ats , time: t >}
}
if   all-busy(SL)    --- [L]::0 is the fst free ag.
/\   o == max(o, S). --- o is the max wrt <

\end{maude}

Let  $\RADTAP$ be as $\RADTA$ but replacing
\code{[time]} and \code{[pick]} with the conditional rules above.

\begin{theorem}[Correctness]
  Let $T$ be an ADTree and suppose that $\RADT$
  finds an attack with time $t$ and number of agents $n$
  using the set of attacks $S$.
  If $\RADTA \vdash [n,t,S,T] \redi^* [SL] $ and $m$ agents in \code{SL} were not assigned
  any task, then $\RADTAP \vdash [n-m,t,S,T] \redi^* [SL']$
  where $SL'$ is as $SL$ but with the $m$ (unused) agents removed.
\end{theorem}
\begin{proof}
  Assume that in a given state, there are two enabled attacks $o$ and $o'$ and
  $o' \prec o$. $\RADTA$ may pick either $o$ or
  $o'$ and $\RADTAP$ is forced to pick $o$.
  We  show that $\RADTA$ necessarily chooses $o$.
  Let $X$ be the common ancestor
  of $o$ and $o'$.
  Since both actions are enabled, $X$ is necessarily an \code{AND}
  gate. The  minimal remaining time $mt$ for $X$ is bound by the maximum
  time needed the perform the actions in the path from $o$ to $X$
  (say $t$) and the time needed to perform
  the path from  $o'$ to $X$ (say $t'$).
  Since $o' \prec o$, then $t' < t$.
  Suppose, to obtain a contradiction, that at a given time,
  $o'$ is  scheduled and $o$ is not. When the time advances,
  $t'$ is decremented but $t$ remains the same. Hence, the deadline
  $mt$ for $X$ cannot be met.
\end{proof}

\subsection{Results}

In the repository \cite{adt2maude} of \toolM, the reader can find the
complete specification of the proposed rewrite theories. A script
written in Python, using the bindings for Maude
(\url{https://github.com/fadoss/maude-bindings}), translates the
input format for ADTrees used in \tool and produces a term representing
the tree ($\os T \cs $). Then, the analyses for finding the minimal
time and the optimal schedule are performed. The resulting schedules
coincide with those reported in \Cref{sec:expe}.
Even though the specialised algorithm outperforms Maude in most cases,
\Cref{table:benchmark} shows that the specification
is useful in practice. Additional benchmarks can be found at \url{https://bit.ly/3ONeSzq}.

\begin{table}
  \centering
  \begin{tabular}{lrr}
    \toprule
    \textbf{model}          & \textbf{\tool (ms)} & \textbf{\toolM (ms)} \\
    \midrule
    forestall               & 20.05                           & 160.12                           \\
    \rowcolor{lightgray}
    gain\_admin             & 3342045.64                      & 6129.03                          \\
    icfem2020               & 2.51                            & 180.35                           \\
    interrupted             & 1.31                            & 121.70                           \\
    \rowcolor{lightgray}
    iot\_dev                & 2652.56                         & 629.17                           \\
    last                    & 1.83                            & 129.84                           \\
    scaling                 & 1.22                            & 122.94                           \\
    treasure\_hunters       & 26.70                           & 131.56                           \\
    adtree-d5\_w3\_c10\_AND & 12.41                           & 7007.98                          \\
    adtree-d5\_w4\_c10\_AND & 13.72                           & 5008.19                          \\
    \bottomrule
  \end{tabular}
  \caption{\tool vs. \toolM in benchmarks\label{table:benchmark}}
\end{table}

Being declarative (since behaviour is easily described by rules) and based on a
search procedure, the rewriting logic specification is easily extensible
to consider other constraints and metrics in ADTrees. For instance, the algorithm (and the
optimisation in $\RADTAP$) assumes that agents can interrupt an activity and
start another one. We may add, as an additional attribute, that such an interruption
requires additional time since the tasks are not executed in the same room. It is also
possible to specify different kind of agents where only some of them are
trained for some specific tasks. The
RL approach also opens the possibility of considering  multi-objective optimisations
including the cost, time and number of agents to perform the attack.

\section{Conclusion}
\label{sec:conclusion}

This paper has presented an agents scheduling algorithm that allows for evaluating
attack/defence models. It synthesises a minimal number of agents and their schedule,
providing insight to both parties as to the number of agents and actions
necessary for a successful attack, and the defences required to counter it.
We have also presented an executable  rewrite theory to solve the same problem.
The specialised algorithm inspired some optimisations that allowed us
to reduce the state space and show that the specification can be used in
practice. 

The declarative model in RL opens different alternatives
to consider other constraints and quantitative measures in \ADT/s. 
We thus obtain a complete framework for not only analysis but also synthesis
of agent configurations and schedules to achieve a given goal in a multi-agent
system. Targeting more elaborate goals, expressed in the TATL logic
\cite{KnapikAPJP19}, will allow for analysing more general multi-agent systems
and their properties.
Also, we plan to use rewriting  modulo SMT \cite{DBLP:journals/jlp/RochaMM17}
to encode configurations induced by \gateOR and defence nodes
and perform symbolic analysis \cite{DBLP:journals/jlap/DuranEEMMRT20} on \ADT/s. 



\bibliographystyle{IEEEtran}
\balance
\bibliography{refs}

\end{document}